\documentclass[onefignum,onetabnum]{siamart190516}

\usepackage{verbatim}

\usepackage{lipsum}
\usepackage{amsfonts}
\usepackage{graphicx}
\usepackage{epstopdf}
\usepackage{algorithmic}
\ifpdf
  \DeclareGraphicsExtensions{.eps,.pdf,.png,.jpg}
\else
  \DeclareGraphicsExtensions{.eps}
\fi

\newsiamremark{remark}{Remark}
\newsiamremark{hypothesis}{Hypothesis}
\crefname{hypothesis}{Hypothesis}{Hypotheses}
\newsiamthm{claim}{Claim}

\headers{Essential spectrum in a short pulse laser}{V. Shinglot and J. Zweck}

\title{The essential spectrum of periodically stationary pulses in  lumped models
of short-pulse fiber lasers\thanks{Submitted to the editors 6/22/2022.
\funding{This work was funded by the NSF under DMS-2106203.
}}}

\author{Vrushaly Shinglot\thanks{Department of Mathematical Sciences, The University of Texas at Dallas, Richardson, TX 75080, USA   (\email{vrushaly.shinglot@utdallas.edu}, \email{zweck@utdallas.edu}, \url{https://personal.utdallas.edu/\string~zweck/}).}\and John Zweck\footnotemark[2]
}

\usepackage{amsopn}

\usepackage{amsmath,amssymb}
\usepackage{physics}
\renewcommand\norm[1]{\left\lVert#1\right\rVert}
\usepackage{csquotes}
\newtheorem*{rem}{Remark}

\newcommand\StartRed{% \StartRed opens group
  \begingroup           % start a group
  \let\EndRed\endgroup % within this 'environment' (only),
}

\ifpdf
\hypersetup{
  pdftitle={An Example Article},
  pdfauthor={D. Doe, P. T. Frank, and J. E. Smith}
}
\fi

\begin{document}

\maketitle

% REQUIRED
\begin{abstract} 
In modern short pulse fiber lasers there is significant pulse breathing over each
round trip of the laser loop. Consequently, averaged models cannot be used for
quantitative modeling and design. Instead, lumped models, 
which are 
obtained by concatenating models for the various components of the laser, are
required. Since the pulses in lumped models are periodic rather than stationary,
their linear stability is evaluated with the aid of the monodromy operator
obtained by linearizing the round trip operator about the periodic pulse. 
Conditions are given on the smoothness and decay of the periodic pulse which
ensure that the monodromy operator exists on an appropriate Lebesgue function space. A formula for the essential spectrum of the monodromy operator is 
given which can be used to quantify the growth rate of continuous wave perturbations. This formula is established by showing that the 
essential spectrum of the monodromy operator equals that of an associated
asymptotic  operator. Since the asymptotic monodromy   operator acts as a 
multiplication operator in the Fourier domain, it is possible to derive a formula
for its spectrum. Although the main results are stated for a particular
experimental stretched pulse laser, the analysis shows that they can be readily
adapted to a wide range of lumped laser models. 
\end{abstract}

% REQUIRED

% REQUIRED

\begin{comment}

Results on existence of psi. See Yannan email. Look at GL paper of weinberg?

ess spec thms for diff ops as in K&P for stat solutions. Our is extension for PS.

\end{comment}

\section{Introduction}\label{sec:Intro}

The purpose of this paper is to establish a formula for the essential spectrum 
of the monodromy operator for a periodic pulse in a lumped model of an experimental 
short pulse fiber laser. The physical importance of the essential spectrum is that it quantifies the growth rate of continuous wave perturbations seeded by quantum mechanical noise in the system. Such 
perturbations can have a major impact on the performance of laser-based systems.
Since the advent of the soliton laser~\cite{mollenauer1984soliton}, researchers have invented 
several generations of short pulse fiber lasers for 
a variety of applications,
including stretched-pulse (dispersion-managed) lasers~\cite{kim2014sub,tamura199377},
similariton lasers \cite{fermann2000self, hartl2007cavity},
and the  Mamyshev oscillator \cite{rochette2008multiwavelength,sidorenko2018self,tarasov2016mode}. 
The pulses in these lasers typically have durations on the order of 100~fs,
peak powers on the order of 1-2~MW, and energy in the 1-50~nJ range.
Applications of femtosecond laser technology include frequency comb generation,
highly accurate measurement of time, frequency, and distance, optical waveform generation, trace-gas sensing, the search for exoplanets, and laser surgery~\cite{diddams2010evolving,fu2018several}.

Traditionally, mathematical modeling and analysis of short pulse lasers has been
based on averaged models,
 in which each of the physical effects that act on the light pulse
is averaged over one round trip of the laser loop to obtain a 
partial differential equation such as the cubic-quintic complex Ginzburg-Landau
equation (CQ-CGLE) or the Haus master equation (see \cite{kutz2006mode} for a review).
This approach has been successfully applied to
soliton lasers for which the pulse  maintains its shape as it propagates
over each round trip. In particular, analytical and computational methods
have been developed to find stationary pulse solutions of these equations
and to analyze their stability using  soliton perturbation theory~\cite{haus1975theory,haus2000mode,kapitula2002stability,kaup1990perturbation,menyuk2016spectral}.
However, as is highlighted in the survey paper of Turitsyn et al.~\cite{turitsyn2012dispersion}, averaged models cannot be used for
the quantitative modeling and design of modern short pulse lasers since
 from one generation of laser to the next there has been
a dramatic increase in the amount by which the pulse varies over each round trip.

Instead, the computational modeling of modern short pulse lasers should 
be based on lumped models obtained by concatenating models for the various
components of the laser.  
Typically short pulse lasers   include an optical fiber amplifier, segments
of  single-mode fiber, a  saturable absorber, a dispersion compensating 
element, a spectral filter, and an output coupler. 
Different laser designs are characterized by different orderings of the components
around the loop and by different sets of physical parameters for each component.
Depending on the modeling goal, the models of the individual components may be
phenomenological or derived from physical laws. 
With a lumped model, the pulse  changes shape as it propagates through the various components of the laser system, returning to the same shape
once per round trip. We call such pulses  \emph{periodically stationary} to distinguish them from the stationary pulses in a soliton laser.

The key goals for the modeling of short pulse lasers are to find parameter regions  in which stationary or periodically stationary solutions exist, determine the stability of these pulses, and within the stability region to optimize the pulse parameters and noise performance for specific applications.

Building on  analytical work of Kaup~\cite{kaup1990perturbation} and Haus~\cite{haus1975theory,haus2000mode},
Menyuk~\cite{menyuk2016spectral}  developed a computational approach to the modeling of stationary pulse solutions of averaged models. 
With this method, stationary pulses are found using a root finding method and their linear stability is determined by computing the spectrum of the
linearization of the governing equation about the pulse. 
(We recall that the spectrum of an operator on a function space
is the union of the essential spectrum and the eigenvalues).
In this context the essential spectrum is elementary to calculate with the aid of 
Weyl's essential spectrum theorem~\cite{kapitula2013spectral}. 
While Menyuk computes the eigenvalues by solving a nonlinear eigenproblem
involving a matrix discretization of the differential operator~\cite{shen2016spectra,wang2014boundary},
 analytical and computational Evans function methods have also been developed
 for the CQ-CGLE and for nonlocal equations
 such as the Haus master equation~\cite{kapitula2004evans,kapitula1998instability,kapitula1998stability}.

Extending this approach to periodically stationary pulses in lumped laser models
is significantly more challenging. In~\cite{shinglot2022continuous}, building on a method of Ambrose and Wilkening
for computing periodic solutions of partial differential equations~\cite{ambrose2012computing},  we developed
 an optimization method to find periodically stationary pulses. Each iteration
of the optimization algorithm involves solving the equations in the model
over one round trip of the laser. In analogy with the Floquet theory of periodic
solutions of ordinary differential equations~\cite{teschl2012ordinary}, we expect that the
linear stability of the resulting periodic pulse will be determined by
the spectrum of the monodromy operator of the linearization of the lumped model
about the pulse. 
Indeed, it should be possible to rigorously establish such a result by generalizing
the Floquet stability theory for parabolic partial
differential equations  developed  by Lunardi~\cite{lunardi2012analytic}. 
In~\cite{shinglot2022continuous} we also presented a formula for the essential spectrum of the monodromy operator and obtained excellent agreement between the formula and a subset of the eigenvalues of a matrix discretization of the operator. This agreement was shown 
for a lumped model of an experimental stretched pulse laser of Kim et al~\cite{kim2014sub}.
The purpose of the current paper is to prove the essential spectrum formula announced 
in~\cite{shinglot2022continuous}.  Our approach builds upon that in Zweck et al.~\cite{zweck2021essential} which dealt with the simpler case of periodically stationary  pulse solutions of the  constant-coefficient CQ-CGLE.

Since we do not yet know how to formulate conditions to ensure that there exists a 
periodically stationary pulse solution to the lumped model, for the results in this paper
we simply assume that the parameters in the model have been chosen so that such a pulse exists. This assumption is reasonable since we have solid numerical evidence for the existence of such pulses~\cite{shinglot2022continuous}. The first main result of the paper, Theorem~\ref{thm:Properties of monodromy operator},
provides conditions on the regularity and decay of the pulse which guarantee
that the monodromy operator exists on an appropriate $L^2$-function space.
Since it is not possible to  calculate the essential spectrum of the monodromy operator directly, we instead compute the essential spectrum of an associated 
asymptotic monodromy operator. The asymptotic operator is defined by taking 
the limit as the spatial variable goes to infinity of the  monodromy operator.
Intuitively, the spectrum of the asymptotic operator provides information about the
growth rate of noise perturbations far from the pulse. The second main result, Theorem~\ref{thm:Theorem for formula of essential spectrum},
is a formula for the essential spectrum of the asymptotic monodromy operator.
This result is established in the Fourier domain, where the asymptotic
operator acts as a multiplication operator on a space of $\mathbb C^2$-valued functions.
The proof relies on a general formula we derive for the spectrum of a
multiplication operator on $L^2(\mathbb R, \mathbb C^2)$. 
The proof of this general formula builds on a similar well known formula in the case
of scalar-valued functions~\cite{engel2001one}, but the case of vector-valued functions
involves some additional technicalities. Finally, in the third main result, Theorem~\ref{thm:EssSpecEqual},
we establish conditions which guarantee  that the essential spectrum of the
monodromy operator equals that of the asymptotic operator. 

To keep the presentation as concrete as possible, rather than attempting to formulate an abstract definition of a general lumped model of a short pulse laser, the theorems are formulated and proved for the Kim laser we modeled in~\cite{shinglot2022continuous}. 
However, based on the discussion at the beginning of this introduction, we 
anticipate that  the results can easily be adapted to most lumped laser models.
In particular, the  formula we derive for the essential spectrum is independent
of the order of the components in the model. Furthermore, provided that the conditions
in the remark following Theorem~\ref{thm:EssSpecEqual} still hold, the 
models for the components can be switched out for
other models, and additional components such as a spectral filter
can be added. Finally, the conditions on the physical parameters we impose
in the main results hold generically.

From a technical point of view there are two  main challenges in extending the results
on the constant coefficient CQ-CGLE in \cite{zweck2021essential}  to lumped laser models. The first
challenge is that 
nonlocality of the gain saturation in the Haus master equation complicates the
proofs of the main theorems.
\StartRed
The physical  implications of the nonlocality of the gain saturation are discussed in Section~\ref{sec:SimResults}. 
\EndRed
 The second challenge is that the monodromy
operator is defined as a composition of solution operators for each component
of the model, which requires adopting a different point of view, especially in the 
proof of the third main result. The combination of these two challenges
ultimately means that the formula for the essential spectrum in the lumped model 
has a different character from the CQ-CGLE case. 

The paper can be outlined as follows. In Section~\ref{sec:Mathematical Model}, we describe the lumped model of the experimental stretched pulse laser of Kim et al.~\cite{kim2014sub} and define the round trip operator, $\mathcal R$. In Section~\ref{sec:Linearization of the Round Trip Operator}, we linearize $\mathcal R$ about a periodically
stationary pulse, $\psi$, to obtain the monodromy operator, $\mathcal M$, and 
the associated asymptotic monodromy operator, $\mathcal M_\infty$.
In Section~\ref{sec:MainResults}, we state the three main theorems of the paper, including formulating the hypotheses on $\psi$  we need to obtain these results.
\StartRed
We also state the formula we derived for the essential spectrum of $\mathcal M$.
In Section~\ref{sec:SimResults} we present some simulation results based on this formula.
\EndRed
In Section~\ref{sec:ExistenceMonodromy}, we 
prove the first main theorem on the existence and regularity properties of $\mathcal M$.
This proof  relies on the concept of an evolution system in semigroup
theory~\cite{pazy2012semigroups} in which linear partial differential equations 
of the form 
$\partial_t \mathbf u = \mathcal L(t) \mathbf u$ are regarded as ordinary differential
equations for trajectories, $t \mapsto \mathbf u(t)$, in an infinite dimensional Banach space. The estimates
in the proof of the technical Lemma~\ref{lem:F is C1} are relegated to Appendix~\ref{AppendixC1}. 
In Section~\ref{sec:Formula for the Spectrum of a Multiplication Operator}, we derive a formula for the spectrum of a general multiplication
operator on $L^2(\mathbb R, \mathbb C^2)$, and in Section~\ref{sec:The Essential Spectrum of the Asymptotic Monodromy Operator} we apply this
formula to calculate the essential spectrum of $\mathcal M_\infty$. In Sections~\ref{sec:Relative compactness for the linearized differential operators in the fiber amplifier} and \ref{sec:analycity}, we prove two theorems  concerning 
the linearized differential operator, $\mathcal L(t)$,
in the fiber amplifier and its asymptotic counterpart,  $\mathcal L_\infty(t)$. The first result states that $\mathcal L(t)$ is a relatively compact perturbation of $\mathcal L_\infty(t)$ and the second result states that the semigroup of the operator $\mathcal L_\infty(t)$ is analytic. Finally, 
these results are used in Section~\ref{sec:EssSpecEqual} to prove the third main
theorem that the essential spectrum of $\mathcal M$ equals the essential spectrum 
of $\mathcal M_\infty$.

\section{Mathematical Model}
\label{sec:Mathematical Model}

In the left panel of Fiigure~\ref{fig:Short pulse laser model}, we show a system diagram 
for the lumped model of the stretched pulse laser of Kim et al.~\cite{kim2014sub}. A light pulse circulates around the loop, passing through a saturable absorber (SA), a segment of single mode fiber (SMF1), a fiber amplifier (FA), a second segment of single mode fiber (SMF2), a dispersion compensation element (DCF), and an output coupler (OC).
\StartRed
After several round trips, the light circulating in the loop forms into a pulse 
that changes shape as it propagates through the different components,
returning to the same shape each time it returns to the same position in the loop.
In the right panel of Figure~\ref{fig:Short pulse laser model} we show the profile
of such a periodically stationary pulse at the output of each component.
The goal of this paper is to study the spectral stability of periodically stationary pulses
in lumped  models of fiber lasers.
\EndRed

\begin{figure}[H]
    \centering
    \includegraphics[width=0.3\linewidth]{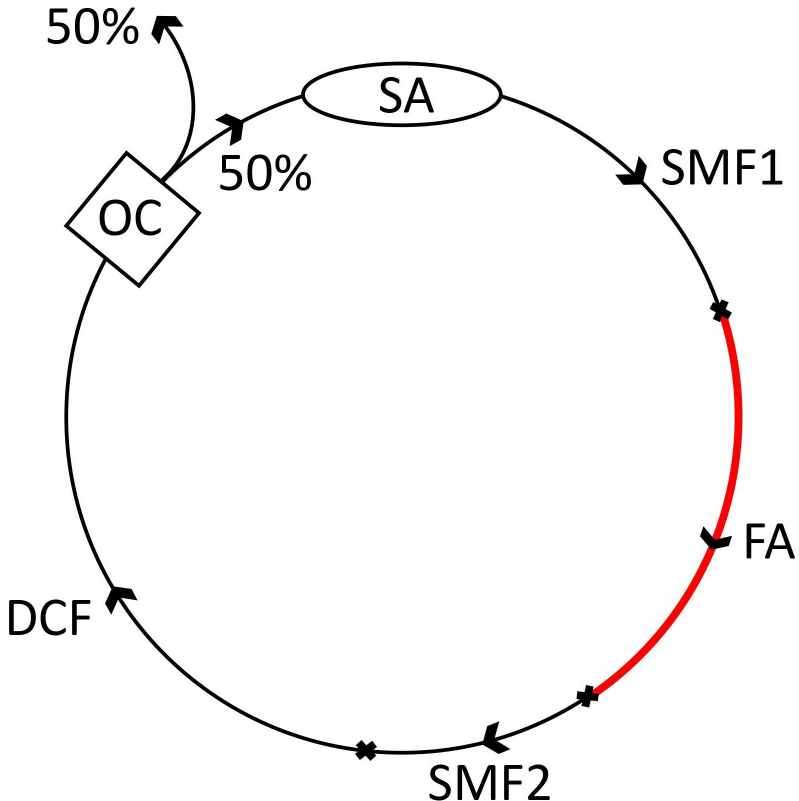}
    \hskip40pt
    \includegraphics[width=0.35\linewidth]{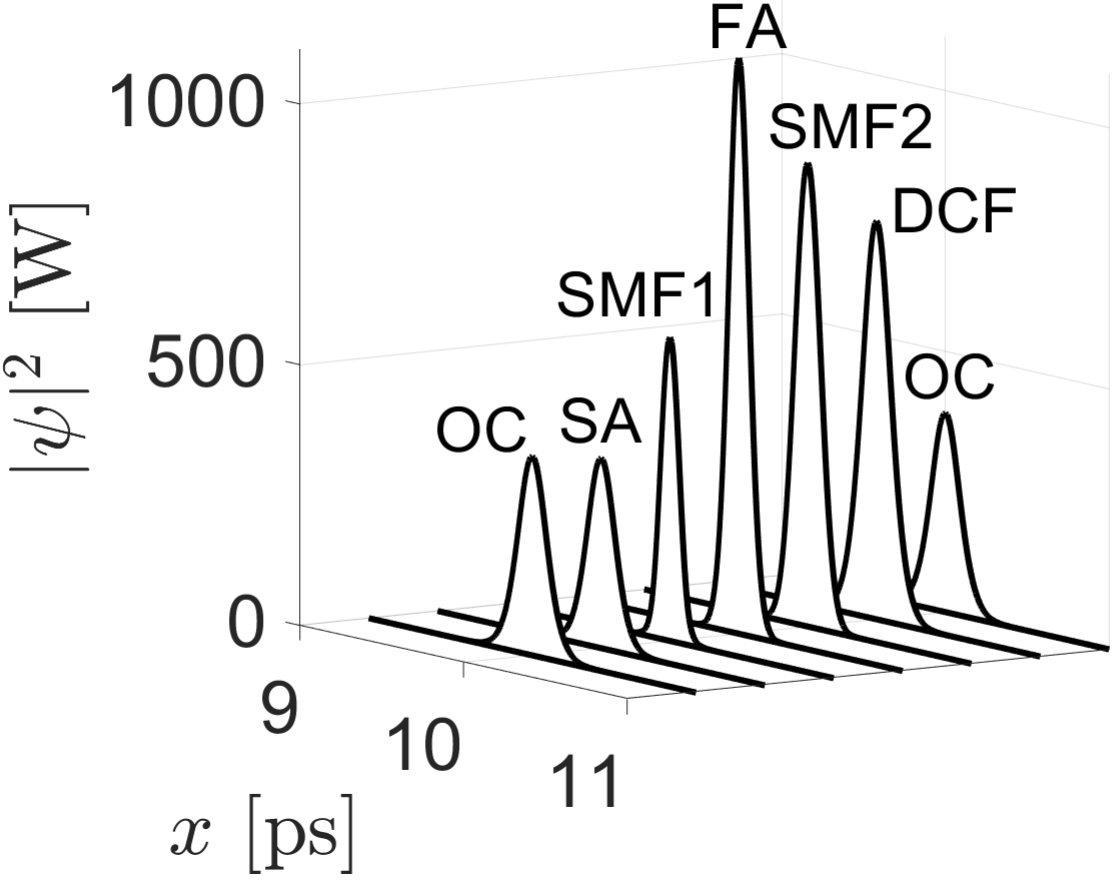}
    \caption{Left: System diagram of the stretched pulse laser of Kim et al.~\cite{kim2014sub}.
    Right: Instantaneous power of the periodically stationary pulse exiting each component of the laser.
    }
    \label{fig:Short pulse laser model}
\end{figure}

At each position in the loop, we model the complex electric field envelope of the light as a function, $\psi = \psi(x)$, 
\StartRed
of a spatial variable, 
$x$, across the pulse. The pulse is normalized so that $|\psi(x)|^2$ is the instantaneous power. 
\EndRed
 We assume that the function, $\psi$, is an element of the Lebesgue space, $L^2(\mathbb{R}, \mathbb{C})$,
of square integrable, complex-valued functions on $\mathbb R$.
We model each component of the laser as a transfer function, 
$\mathcal{P}:  L^2(\mathbb{R}, \mathbb{C})\rightarrow L^2(\mathbb{R}, \mathbb{C})$, so that
\begin{equation}
    \psi_{\text{out}} = \mathcal{P} \psi_{\text{in}},
    \label{eq:Transfer function}
\end{equation}
where $\psi_{\text{in}}$ and $\psi_{\text{out}}$ are the pulses entering and exiting the component. 
\StartRed
The components in the model come in two flavors: discrete and continuous.  
By a discrete component we mean one in which the action of the operator, 
$\mathcal P$, on the input pulse, $\psi_{\text{in}}$,  is  essentially
obtained in one step, for example by the application of an 
explicit formula. In our model of the Kim laser,  the discrete components are the
saturable absorber, dispersion compensation  element,  and  output coupler.
Short-pulse fiber lasers sometimes also  include a spectral filter
that is modeled as a discrete component. 
By a  continuous component, we mean one in which the action 
of the operator, $\mathcal P$, on the input pulse, $\psi_{\text{in}}$,
is modeled by solving a nonlinear wave equation with initial condition,
$\psi_{\text{in}}$, from the input to the output of the component.
In fiber lasers the continuous components are those 
that involve the propagation of a light pulse through 
a segment of nonlinear optical fiber. For our model of the Kim laser these are
the fiber amplifier and the two segments of single mode fiber.
Note that we have chosen to model the dispersion compensation  element as 
a discrete component, since it is modeled by a constant-coefficient 
linear partial differential equation  which has an analytical solution in the Fourier domain.

With a lumped model, the propagation of a light pulse once around the laser loop
is modeled by the round trip operator, $\mathcal{R}: L^2(\mathbb{R}, \mathbb{C}) \rightarrow L^2(\mathbb{R}, \mathbb{C})$, which is given by the 
composition of the transfer functions of all the components. 
For our model of the Kim laser,  the round trip operator is given by 
\EndRed
\begin{equation}
    \mathcal{R} = \mathcal{P}^{\text{OC}} \circ \mathcal{P}^{\text{DCF}}\circ\mathcal{P}^{\text{SMF2}} \circ\mathcal{P}^{\text{FA}}\circ\mathcal{P}^{\text{SMF1}} \circ \mathcal{P}^{\text{SA}}.
    \label{eq:Round trip operator}
\end{equation}
\StartRed
We say that $\psi_0 \in L^2(\mathbb{R}, \mathbb{C})$ 
\EndRed
is a \emph{periodically stationary pulse} if 
\begin{equation}
 \mathcal{R}(\psi_0) = e^{i\theta} \psi_0,
\end{equation}
\StartRed
for some constant phase, $\theta \in [0,2\pi )$. 
For the Kim laser, $\psi_0$ is the pulse at the input to the saturable absorber. For each component, we let $\psi_{\text{in}}$ denote the pulse obtained by propagating the periodically stationary pulse, $\psi_0$, from 
the input to the SA  to the input to that component.
For the continuous fiber 
components we let $\psi$ denote the pulse propagating through that fiber. 
\EndRed
 In~\cite{shinglot2022continuous}, we formulated the problem of discovering periodically stationary pulses as that of finding a zero of the Poincar\'e map functional, 
 \StartRed
 $\mathcal{E}:L^2(\mathbb{R}, \mathbb{C}) \times [0,2\pi ) \rightarrow\mathbb{R}$, 
 \EndRed
 given by
\begin{equation}
    \mathcal{E}(\psi_0, \theta) = \frac{1}{2} \norm{\mathcal{R}(\psi_0) - e^{i\theta} \psi_0}_{L^2(\mathbb{R}, \mathbb{C})}^2.
    \label{eq:Poincare map functional}
\end{equation}
Since $\mathcal{E}\geq 0$, in practice we minimize $\mathcal{E}$ with respect to $\psi_0$ and $\theta$ using a gradient-based iterative optimization method. 
In the right panel of Figure~\ref{fig:Short pulse laser model}, we 
plot the optical power of a periodically stationary pulse obtained using this method.

\StartRed
We now describe the model for the propagation of  a light pulse, $\psi= \psi(t,x)$, through the fiber amplifier. 
Here $t$ denotes position along the fiber, with $0\leq t \leq L_{\text{FA}}$, where $L_{\text{FA}}$ is the length of the fiber
amplifier.  We note that  $t$ is a local evolution variable that is 
only defined within the fiber amplifier. 
Mathematically, we regard   $x$ as the spatial
variable across the pulse. Physically speaking, it is a fast time variable
defined relative to a frame moving at the group
velocity~\cite{yang2010nonlinear}. 
\EndRed
Our model for propagation in the fiber amplifier is based on the Haus master equation~\cite{haus1975theory}, which is a generalization of the nonlinear Schr\'odinger equation that includes gain that saturates at high energy and is of finite bandwidth. Specifically, we model the transfer function, $\mathcal{P}^{\text{FA}}$, of a fiber amplifier of length, $L_{\text{FA}}$, as
 $   \psi_{\text{out}} = \mathcal{P}^{\text{FA}}\psi_{\text{in}}$,
 where $\psi_{\text{out}} = \psi({L}_{\text{FA}},\cdot)$ is  obtained by solving the initial value problem
 \StartRed
\begin{equation}
\begin{aligned}
    \partial_t \psi 
    &= \left[ \frac{ g(\psi)}{2} \left(1+ \frac{1}{\Omega_g^2}\partial_x^2
    \right) 
    - \tfrac{i}{2}\beta_{\text{FA}} \partial_x^2 
    + i\gamma|\psi|^2\right]\,\psi,
    \qquad \text{for } 0\leq t \leq L_{\text{FA}},
    \\
    \psi(0,\cdot) &= \psi_{\text{in}}.
    \label{eq:fiber amplifier}
\end{aligned}
\end{equation}
Here, $g(\psi)$ is the saturable gain given by
\begin{equation}
    g(\psi) {}={} \frac{g_0}{1+E(\psi)/E_{\text{sat}}},
    \label{eq:Saturable gain}
\end{equation}
where $g_{0}$ is the unsaturated gain, $E_{\text{sat}}$ is the saturation energy, and $E(\psi)$ is the pulse energy, which is given by 
\begin{equation}
    E(\psi) = \int_{\mathbb{R}} |\psi(\cdot,x)|^2 dx.
\end{equation}
We note that the energy, and hence the saturable gain, are nonlocal in the spatial variable, $x$, and that they depend on the evolution variable, $t$, since $\psi$ does.
\EndRed
The finite bandwidth of the amplifier is modeled using a Gaussian filter with bandwidth, $\Omega_{g}$. In \eqref{eq:fiber amplifier}, $\beta_{\text{FA}}$ is the chromatic dispersion coefficient and $\gamma$ is the nonlinear Kerr coefficient.

\StartRed
Similarly, we model the transfer function, $\mathcal{P}^{\text{SMF}}$, of a segment of single mode fiber of length, $L_{\text{SMF}}$, as
 $   \psi_{\text{out}} = \mathcal{P}^{\text{SMF}}\psi_{\text{in}}$,
 where $\psi_{\text{out}} = \psi({L}_{\text{SMF}},\cdot)$ is  obtained by solving the initial value problem for the nonlinear Schr\"odinger equation given by
\begin{equation}
\begin{aligned}
    \partial_t \psi 
    &= 
    - \tfrac{i}{2}\beta_{\text{SMF}}  \partial_x^2 \psi
    + i\gamma|\psi|^2 \psi,
    \qquad \text{for } 0\leq t \leq L_{\text{SMF}},
    \\
    \psi(0,\cdot) &= \psi_{\text{in}}.
    \label{eq:SMF_JZ}
\end{aligned}
\end{equation}
We model the dispersion compensation element as $\mathcal P_{\text{DCF}}
= \mathcal F^{-1} \circ \widehat{\mathcal P}^{\text{DCF}} \circ \mathcal F$, where 
$\mathcal F$ is the Fourier transform and 
 \begin{equation}
 \widehat{\psi}_{\text{out}} (\omega) =    (\widehat{\mathcal{P}}^{\text{DCF}} \widehat{\psi}_{\text{in}})(\omega) = 
    \exp \left( i {\omega^2} \beta_{\text{DCF}}/2  \right)
    \widehat{\psi}_{\text{in}}(\omega),
    \label{eq:DCF_JZ}
\end{equation}
with $\widehat{\psi} = \mathcal F (\psi)$.
We observe that \eqref{eq:DCF_JZ} is the solution to the  initial value problem for the
linear equation obtained by setting $\gamma = 0$, $\beta_{\text{SMF}}=
\beta_{\text{DCF}}$ and $L_{\text{SMF}}=1$ in \eqref{eq:SMF_JZ}.

\EndRed

We model the saturable absorber using the fast saturable loss transfer function~\cite{wang2016comparison}, $\mathcal{P}^{\text{SA}}$, given by
\begin{equation}
    \psi_\text{out} = \mathcal{P}^{\text{SA}}(\psi_{\text{in}}) = \left( 1 - \frac{\ell_0}{1+|\psi_\text{in}|^2/ P_{\text{sat}}}\right)\psi_\text{in},\\
    \label{eq:Saturable absorber}
\end{equation}
where $\ell_0$ is the unsaturated loss and $P_{\text{sat}}$ is the saturation power. With this model,  $\psi_\text{out}$ at  $x$  only depends on $\psi_\text{in}$ at the same value of $x$. Finally, we model the transfer function, $\mathcal{P}^{\text{OC}}$, of the output coupler as 
\begin{equation}
    \psi_{\text{out}} = \mathcal{P}^{\text{OC}}\psi_{\text{in}}
    = \ell_{\operatorname{OC}}\,\psi_{\text{in}},
    \label{eq:Output coupler}
\end{equation}
where  $ (\ell_{\operatorname{OC}})^2$ is the power loss at the output coupler.

%---------------------------------------------------------------------------

%---------------------------------------------------------------------------
\section{Linearization of the Round Trip Operator}
\label{sec:Linearization of the Round Trip Operator}

In this section, we derive the 
\StartRed
linearizations, $\mathcal{U}$, about a  pulse 
of each of the  operators, 
$\mathcal{P}$, defined in Section~\ref{sec:Mathematical Model}. 
By the chain rule, 
\EndRed 
the linearization, $\mathcal{M}$, of the round trip operator, $\mathcal{R}$, 
\StartRed
about a periodically stationary pulse, $\psi_0$,
\EndRed
is equal to the composition of the linearized transfer functions, $\mathcal{U}$, i.e., 
\StartRed
\begin{equation}
    \mathcal{M} = \mathcal{U}^{\text{OC}} \circ \mathcal{U}^{\text{DCF}} \circ \mathcal{U}^{\text{SMF2}} \circ \mathcal{U}^{\text{FA}} \circ \mathcal{U}^{\text{SMF1}} \circ \mathcal{U}^{\text{SA}}.
    \label{eq:Linearization of the round trip operator}
\end{equation}
In analogy with the monodromy matrix  in the Floquet theory of periodic solutions of ODE's~\cite{teschl2012ordinary}, we call  $\mathcal M$ the \emph{monodromy operator} of the linearization
of the round trip operator, $\mathcal R$,
 about the periodically stationary pulse, $\psi_0$. 
\EndRed

Because the linearization of the partial differential equations in the model involves the complex conjugate of $\psi$, we reformulate the model as a system of equations for the column vector $\boldsymbol{\psi} = [\Re(\psi), \Im(\psi)]^T \in \mathbb R^2$. 
For example, the transfer function of the fiber amplifier
is reformulated as the operator,
 $\mathcal{P}^{\text{FA}} :  L^2(\mathbb{R}, \mathbb{R}^2) \to  L^2(\mathbb{R}, \mathbb{R}^2)$, given by
$ \boldsymbol{\psi}_{\text{out}} = \mathcal{P}^{\text{FA}}\boldsymbol{\psi}_{\text{in}}$, 
where $\boldsymbol{\psi}_{\text{out}} = \boldsymbol{\psi}(\text{L}_{\text{FA}},\cdot)$ is  obtained by solving the initial value problem
\begin{equation}
\begin{aligned}
    \partial_t \boldsymbol{\psi} 
    &= \left[ \tfrac{ g(\boldsymbol{\psi})}{2} \left(1+ \tfrac{1}{\Omega_g^2}\partial_x^2
    \right) 
    - \tfrac{\beta}{2} \textbf{J} \partial_x^2 
    + \gamma \norm{\boldsymbol{\psi}}^2 \textbf{J} \right]\,\boldsymbol{\psi},
    \\
    \boldsymbol{\psi}(0,\cdot) &= \boldsymbol{\psi}_{\text{in}},
    \label{eq:fiber amplifier new}
\end{aligned}
\end{equation}
where $\textbf{J} = \begin{bmatrix} 0&-1\\1&0 \end{bmatrix}$, and $\norm{\cdot}$ is the standard Euclidean norm on $\mathbb{R}^2$.

The linearized transfer function, $\mathcal{U}^{\text{FA}}:  L^2(\mathbb{R}, \mathbb{R}^2) \to  L^2(\mathbb{R}, \mathbb{R}^2)$, in the fiber amplifier 
is given by $\boldsymbol{u}_{\text{out}} = \mathcal{U}^{\text{FA}}\boldsymbol{u}_{\text{in}}$, 
where $\boldsymbol{u}_{\text{out}} = \boldsymbol{u}({L}_{\text{FA}},\cdot)$ is obtained by solving the linearized initial value problem
\begin{equation}
    \begin{aligned}
    \partial_t \boldsymbol{u} &= \left[g(\boldsymbol{\psi})\textbf{K} + \textbf{L} + \textbf{M}_1(\boldsymbol{\psi}) + \textbf{M}_2(\boldsymbol{\psi})\right] \boldsymbol{u} + \textbf{P}(\boldsymbol{\psi}, \boldsymbol{u}), 
    \qquad \text{for } 0\leq t \leq L_{\text{FA}}
    \\
    \boldsymbol{u}(0,\cdot) &= \boldsymbol{u}_{\text{in}},
    \end{aligned}
    \label{eq:FA linearized equation}
\end{equation}
where
\begin{equation}
    \begin{aligned}
    \textbf{K} &= \tfrac{1}{2}\left(1 + \tfrac{1}{\Omega_{g}^{2}}\partial_x^2\right), 
    \qquad 
    &\textbf{L} &= -\tfrac{\beta}{2}\textbf{J}\partial_{x}^2,\\
    \textbf{M}_1(\boldsymbol{\psi}) &= \gamma\norm{\boldsymbol{\psi}}^2\textbf{J},
    \qquad 
    &\textbf{M}_2(\boldsymbol{\psi}) &= 2\gamma\textbf{J}\boldsymbol{\psi}\boldsymbol{\psi}^T,
    \end{aligned}
    \label{eq:Operators in linearized equation}
\end{equation}
and 
\begin{equation}
\textbf{P}(\boldsymbol{\psi}, \boldsymbol{u}) = -\tfrac{g^{2}(\boldsymbol{\psi})}{g_{0}E_{\text{sat}}}\left[\left(1 + \tfrac{1}{\Omega_{g}^{2}}\partial_x^2\right)\boldsymbol{\psi}\right]\int_{-\infty}^{\infty} \boldsymbol{\psi}^T(x) \boldsymbol{u}(x) dx
\label{eq:Discretization of operator P}
\end{equation}
is a nonlocal operator. 
The non-locality of $\mathbf P$, which arises because the gain saturation depends on the total energy of the pulse, makes the analysis more challenging for the fiber
amplifier than for a segment of single mode fiber.
\StartRed
The linearized transfer function, $\mathcal{U}^{\text{SMF}}$, of a segment of single mode fiber  is obtained by setting $g(\boldsymbol\psi) = 0$  in
\eqref{eq:FA linearized equation} and \eqref{eq:Discretization of operator P}.
\EndRed

The linearized transfer function, $\mathcal{U}^{\text{SA}}$,  for the saturable absorber is given by
\begin{equation}
    \boldsymbol{u}_{\text{out}} = \mathcal{U}^{\text{SA}}(\boldsymbol{\psi}_{\text{in}}) \boldsymbol{u}_{\text{in}} = \left( 1 - \ell(\boldsymbol{\psi}_{\text{in}}) - \frac{2 \ell^2(\boldsymbol{\psi}_{\text{in}})}{\ell_0 P_{\text{sat}}} \boldsymbol{\psi}_{\text{in}} \boldsymbol{\psi}_{\text{in}}^T \right) \boldsymbol{u}_{\text{in}},
    \label{eq:SA linearized transfer function}
\end{equation}
where
\begin{equation}
    \ell(\boldsymbol{\psi}_\text{in}) = \frac{\ell_0}{1+\norm{\boldsymbol{\psi}_{\text{in}}}^2/ P_{\text{sat}}}.
    \label{eq:l(psi) in SA transfer function}
\end{equation}
\StartRed
The remaining components, i.e. dispersion compensation fiber and output coupler, already have linear transfer functions, and so $\mathcal U^{\text{DCF}} = \mathcal P^{\text{DCF}}$ and $\mathcal U^{\text{OC}} = \mathcal P^{\text{OC}}$.
\EndRed

Because eigenvalues and eigenfunctions can be complex valued, we extend the linearized system to act on complex-valued functions, $\boldsymbol{u} \in L^2(\mathbb{R}, \mathbb{C}^2)$, where 
\begin{equation}
    L^2(\mathbb{R}, \mathbb{C}^2) = \{ \boldsymbol{u} = \boldsymbol{v} + i \boldsymbol{w} \,\,: \,\, \boldsymbol{v}, \boldsymbol{w} \in L^2(\mathbb{R}, \mathbb{R}^2) \},
\end{equation}
is the space of $\mathbb{C}^2$-valued functions on $\mathbb{R}$ with the standard Hermitian inner product.
Let $\mathcal T$ be an operator that acts on  $\mathbb{R}^2$-valued functions.
We extend $\mathcal T$ to act on $\mathbb{C}^2$-valued functions by defining
$  \mathcal{T} \boldsymbol{u} = \mathcal{T} \boldsymbol{u}_1 + i \mathcal{T} \boldsymbol{u}_2$.
where $\boldsymbol{u} = \boldsymbol{u}_1 + i \boldsymbol{u}_2$ with $\boldsymbol{u}_1, \boldsymbol{u}_2 \in L^2(\mathbb{R}, \mathbb{R}^2)$. 
Note that the formulae above 
for the action of the differential operators and transfer functions 
on $\mathbb{C}^2$-valued functions, $\boldsymbol{u}$,
are the same as for their action
on $\mathbb{R}^2$-valued functions, since in both cases we only require
$\boldsymbol\psi$ to be  $\mathbb{R}^2$-valued. The only difference is our
interpretation of the function spaces on which they act.

The linear stability of the pulse $\psi$ is determined by the spectrum of the monodromy operator, $\mathcal{M}$, which is the union of the essential spectrum of $\mathcal{M}$ and the eigenvalues of $\mathcal{M}$. In Section~\ref{sec:MainResults}, we show that the essential spectrum of the monodromy operator is equal to the essential spectrum of an associated  \emph{asymptotic monodromy operator}, $\mathcal{M}_{\infty}$, 
which is defined by 
\StartRed
\begin{equation}
    \mathcal{M}_{\infty} = \mathcal{U}_{\infty}^{\text{OC}} \circ \mathcal{U}_{\infty}^{\text{DCF}} \circ \mathcal{U}_{\infty}^{\text{SMF2}} \circ \mathcal{U}_{\infty}^{\text{FA}} \circ \mathcal{U}_{\infty}^{\text{SMF1}} \circ \mathcal{U}_{\infty}^{\text{SA}},
    \label{eq:Asymptotic linearized roundtrip operator}
\end{equation}
\EndRed
where each operator, $\mathcal{U}_{\infty}$, is the $x$-independent operator obtained by taking the limit as $\abs{x} \to \infty$ of the corresponding operator, $\mathcal{U}$. 
\StartRed
In Section~\ref{sec:MainResults}, we will impose conditions on the pulse that ensure that
these limits exist. Under these conditions, each operator 
\EndRed
$\mathcal{U}_{\infty}$ can be obtained by setting $\boldsymbol\psi = 0$ in the corresponding formula for $\mathcal{U}$. We refer to the operators, $\mathcal{U}_{\infty}$, as \emph{asymptotic linearized transfer functions}.

%---------------------------------------------------------------------------
\section{Main Results}
\label{sec:MainResults}

In this section, we first state a theorem that establishes the existence, uniqueness, and regularity properties of the monodromy operator, $\mathcal{M}$, given by \eqref{eq:Linearization of the round trip operator}. 
Essentially the same result also holds for the asymptotic monodromy operator,
 $\mathcal{M}_\infty$, given by \eqref{eq:Asymptotic linearized roundtrip operator}.
Then we provide an explicit formula for the essential spectrum of $\mathcal{M}_\infty$.
The last major result is a theorem stating that 
essential spectrum of $\mathcal{M}$ equals that of $\mathcal{M}_\infty$. 
 
Rigorously proving the existence, uniqueness, and regularity of periodically stationary pulse solutions, ${\psi}$, of the lumped model is challenging. Instead, for the results in this paper, we assume that a periodically stationary pulse, ${\psi}$, exists. This assumption is reasonable since we have strong numerical evidence for the existence of such pulses~\cite{shinglot2022continuous}. We do however need to impose some regularity and decay hypothesis on ${\psi}$ to guarantee the existence of a monodromy operator for ${\psi}$ and to prove the results about the essential spectrum. 
These can be stated as follows.

\begin{hypothesis}
    \label{hyp:Conditions on the solution of SA}
    The pulse, ${\psi}_{\text{in}}$, about which the transfer function,
    \eqref{eq:Saturable absorber},  of the saturable absorber  is linearized has the property that ${\psi}_{\text{in}}$, $\partial_x {\psi}_{\text{in}}$, and $\partial_x^2 {\psi}_{\text{in}}$ are bounded and continuous on $\mathbb{R}$, and ${\psi}_{\text{in}}$ decays exponentially to zero as $x \to \pm \infty$.
\end{hypothesis}

\begin{hypothesis}
\label{hyp:Conditions on the solution of SMF}
\StartRed
The pulse, ${\psi}$, about which equation \eqref{eq:SMF_JZ} for each single mode fiber of length, $L_{\text{SMF}}$, 
\EndRed
is linearized  has the following properties:
\begin{enumerate}
 \item[(a)] ${\psi}$, $\partial_t {\psi}$ are continuous in $t$, uniformly in $x$;
    \item[(b)] For each $t$, the function ${\psi}(t,\cdot) \in L^{\infty}(\mathbb{R}, \mathbb{C})$;
    \item[(c)] For each $t$, the weak derivative $\partial_x{\psi}(t,\cdot) \in L^{\infty}(\mathbb{R}, \mathbb{C})$;
    \item[(d)] There exist constant $r > 0$  so that
    \begin{equation}
        \lim_{{x}\to \pm\infty} e^{r \abs{x}} |{{\psi}(t,x)}| = 0, \quad \text{for all } t
        \StartRed
        \in [0, L_{\text{SMF}}].
        \EndRed
    \end{equation}
\end{enumerate}
\end{hypothesis}

\begin{hypothesis}
    \label{hyp:conditions on the solution of FA}
    In the fiber amplifier 
    \StartRed of length, $L_{\operatorname{FA}}$,  \EndRed
    the pulse, ${\psi}$, about which \eqref{eq:fiber amplifier} is linearized has the same properties as in  Hypothesis~\ref{hyp:Conditions on the solution of SMF},
    in addition to which
    \begin{enumerate}
        \item[(a)] For almost all $x \in \mathbb{R}$, ${\psi}$ is $C^2$ in $t$;
        \item[(b)] For almost all $x \in \mathbb{R}$, $\partial_x^2 {\psi}$, $\partial_t (\partial_x {\psi})$, $\partial_t (\partial_x^2 {\psi})$ are continuous in $t$;
        \item[(c)] There exists $h \in L^2(\mathbb{R}, \mathbb{R}) \cap L^{\infty}(\mathbb{R}, \mathbb{R})$ so that 
        \begin{equation}
            \abs{\partial_t^{(k)} \partial_x^{(\ell)} {\psi}(t,x)} \leq h(x) \quad \text{for } k = 0,1, \; \ell=0,1,2,
        \end{equation}
        and
        \begin{equation}
            \abs{\partial_t^2 {\psi}(t,x)} \leq h(x),
        \end{equation}
        \StartRed
        for all $t\in [0, L_{\operatorname{FA}}]$ 
        \EndRed
        and almost all $x \in \mathbb{R}$.
    \end{enumerate}
\end{hypothesis}

\begin{rem}
Property (c) of Hypothesis~\ref{hyp:conditions on the solution of FA} holds if all the functions $\partial_t^{(k)} \partial_x^{(\ell)} {\psi}$ are bounded  and decay exponentially as in property (d) of Hypothesis~\ref{hyp:Conditions on the solution of SMF}.
\end{rem}

 Let $\mathcal{B}(X)$ denote the space of bounded linear operators on a Banach space, $X$. 
Then we have the following theorem on the existence, unqiueness,
and regularity of the monodromy operator.

\begin{theorem}
\label{thm:Properties of monodromy operator}
\StartRed
Let $\boldsymbol{\psi}_0$ be a 
periodically stationary solution of the lumped laser model, i.e., a solution 
of  \eqref{eq:Round trip operator}. 
\EndRed
Under Hypotheses~\ref{hyp:Conditions on the solution of SA},  \ref{hyp:Conditions on the solution of SMF}, and \ref{hyp:conditions on the solution of FA}, the monodromy operator, $\mathcal{M}$, 
\StartRed
in \eqref{eq:Linearization of the round trip operator}, which is the linearization of  
the round trip operator, $\mathcal{R}$, about $\boldsymbol{\psi}_0$, 
\EndRed
has the following properties:
\begin{enumerate}
    \item[(a)] $\mathcal{M} \in \mathcal{B}(L^2(\mathbb{R}, \mathbb{C}^2))$;
    \item[(b)]  $\mathcal{M}(H^2(\mathbb{R}, \mathbb{C}^2)) \subset H^2(\mathbb{R}, \mathbb{C}^2)$;
    \item[(c)]  For each $\boldsymbol{v} \in H^2(\mathbb{R}, \mathbb{C}^2)$, $\boldsymbol{u} = \mathcal{M}(\boldsymbol{v})$ is the unique solution after one round trip of the linearization of $\mathcal{R}$ about $\boldsymbol{\psi}$.
\end{enumerate}
\end{theorem}

\begin{rem}
An analgous result  holds for the asymptotic monodromy operator, $\mathcal{M}_{\infty}$, given by \eqref{eq:Asymptotic linearized roundtrip operator}.
\end{rem}

Next, we recall the definition of the essential spectrum used in the results below.

\begin{definition}
\label{d:Spectra}
Let $\mathcal{A}:D(\mathcal{A}) \subset X \to X$ be a linear operator with domain, 
$D(\mathcal{A})$, on a Banach space, $X$. We suppose that $\mathcal A$ is closed and densely defined. The \emph{resolvent set} of $\mathcal{A}$ is 
\begin{equation}
    \rho(\mathcal{A}) = \{ \lambda \in \mathbb{C} \, : \, \mathcal{A} - \lambda \text{ is invertible and } (\mathcal{A} - \lambda)^{-1} \in \mathcal{B}(X) \},
\end{equation}
and for each $\lambda \in \rho(\mathcal{A})$, the \emph{resolvent operator} is $\mathcal{R}(\lambda : \mathcal{A}) = (\mathcal{A} - \lambda)^{-1}$.
 The \emph{spectrum} of $\mathcal{A}$ is $\sigma(\mathcal{A}) = \mathbb{C} \backslash \rho(\mathcal{A})$. The \emph{point spectrum} of $\mathcal{A}$ is 
\begin{equation}
    \sigma_{\emph{pt}}(\mathcal{A}) = \{ \lambda \in \mathbb{C} \, : \, \emph{Ker}(\mathcal{A} - \lambda) \neq \{0\} \}.
\end{equation}
The \emph{Fredholm point spectrum} of $\mathcal{A}$ is the subset of $\sigma_{\emph{pt}}(\mathcal{A})$ defined by
\begin{equation}
    \sigma_{\emph{pt}}^{\mathcal{F}}(\mathcal{A}) = \{ \lambda \in \mathbb{C} \,\,:\,\,\mathcal{A} - \lambda \text{ is Fredholm, }\emph{Ind}(\mathcal{A} - \lambda) = 0 \text{ and } \emph{Ker}(\mathcal{A} - \lambda) \neq \{0\} \},
\end{equation}
and the \emph{essential spectrum} of $\mathcal{A}$ is $\sigma_{\emph{ess}}(\mathcal{A}) = \sigma(\mathcal{A}) \backslash \sigma_{\emph{pt}}^{\mathcal{F}}(\mathcal{A})$. 
\end{definition}

\StartRed
\begin{rem}
Although $\sigma(\mathcal{A}) = \sigma_{\emph{pt}}(\mathcal{A}) \cup \sigma_{\emph{ess}}(\mathcal{A})$,  this union may not be disjoint.
\end{rem}
\EndRed

\begin{rem}
There are several inequivalent definitions of the essential spectrum of a closed and densely defined operator. Here, we use the same definition of the essential spectrum as in Zweck et al.~\cite{zweck2021essential}.
This definition  is chosen so that $\sigma_{\text{ess}}(\mathcal{A})$ is the largest subset of the spectrum of $\mathcal{A}$ that is invariant under compact 
perturbations~\cite{edmunds2018spectral}.
\end{rem}

Next, we state a formula for the essential spectrum of $\mathcal{M}_{\infty}$. 
This formula involves the
total dispersion in one round trip of the laser system, which for the stretched pulse laser
is given by $\beta_{\text{RT}} = \beta_{\text{SMF1}}\text{L}_{\text{SMF1}} + \beta_{\text{FA}}\text{L}_{\text{FA}} + \beta_{\text{SMF2}}\text{L}_{\text{SMF2}} + \beta_{\text{DCF}}$. 
\StartRed
Here  
$\beta_{\text{FA}}$, $\beta_{\text{SMF}}$, and $\beta_{\text{DCF}}$, are
the dispersion parameters given in \eqref{eq:fiber amplifier}, \eqref{eq:SMF_JZ}, and \eqref{eq:DCF_JZ}, respectively.
\EndRed

\begin{theorem}
\label{thm:Theorem for formula of essential spectrum}
\StartRed
Suppose that the hypotheses of Theorem~\ref{thm:Properties of monodromy operator} hold, and that
\EndRed
 $\ell_0 \neq 1$ and either $(i)$ $\beta_{\text{RT}} \neq 0$ or $(ii)$ $\Omega_g < \infty$ and $\int_{0}^{L_{\operatorname{FA}}} g(\psi(t))dt \neq 0$. Then the essential spectrum of the asymptotic monodromy operator, $\mathcal{M}_{\infty}$, in 
\eqref{eq:Asymptotic linearized roundtrip operator}
is given by 
\begin{equation}
    \sigma_{\emph{ess}}(\mathcal{M}_{\infty}) = \sigma(\mathcal{M}_{\infty}) = \{\, \lambda_{\pm}(\omega) \in \mathbb{C} \,\, | \,\, \omega \in \mathbb{R} \,\} \cup \{ 0 \},
    \label{eq:Formula for essential spectrum new}
\end{equation}
where
\begin{equation}
    \lambda_{\pm}(\omega) = \ell_{\operatorname{OC}}
    (1-\ell_0) \exp\left\{ \frac 12 \left( 1 - \frac{\omega^2}{\Omega_g^2} \right)\int_0^{L_{\operatorname{FA}}} g(\psi(t))dt \right\} \exp\left\{ \pm i\frac{\omega^2}2\beta_{\text{RT}} \right\}.
    \label{eq:Formula for eigenvalues in essential spectrum}
\end{equation}
\end{theorem}

\begin{rem}
\Cref{eq:Formula for eigenvalues in essential spectrum} can be readily adapted to other lumped fiber 
laser models, provided that
formulae can be found for the  Fourier transforms of all the asymptotic linearized
transfer functions, $\mathcal U_\infty$, in the model. In particular, the formula
is independent of the order in which the components are arranged around the loop.\end{rem}

\StartRed
To prove that the essential spectrum of $\mathcal M$ equals that of $\mathcal M_\infty$
we require that the  linearization of the equation modeling 
the single mode fiber segments (SMF1  and SMF2) generates an  analytic semigroup. 
To do so, we  add an additional spectral filtering term to the 
nonlinear Schr\"odinger equation, so that light propagation in these fibers is modeled by
\EndRed
\begin{equation}\label{eq:NLSFilter}
\partial_t\psi = -\frac{i}{2} \beta \partial_x^2\psi + i\gamma |\psi|^2\psi + \epsilon \partial_x^2\psi,
\end{equation}
where the parameter, $\epsilon$, is  required to be positive, but can be arbitrarily small. 
\StartRed
Provided that $\epsilon>0$, the semigroup for the linearized equation is analytic (see Section~\ref{sec:analycity}).
\EndRed
In the frequency domain the additional term corresponds to
\begin{equation}
\partial_t\widehat\psi(\omega) = - \epsilon \omega^2 \widehat\psi(\omega),
\end{equation}
which models a frequency-dependent loss. The addition of this term is physically reasonable since the loss in optical fiber is wavelength dependent with a minimum at about 1550~nm~\cite{agrawal2007nonlinear}.

\begin{theorem}\label{thm:EssSpecEqual}
Suppose that the hypotheses of Theorem~\ref{thm:Properties of monodromy operator} hold,
and  that in the fiber amplifier
$0< \Omega_g < \infty$ and $(g_0,\beta)\neq(0,0)$. 
Furthermore, suppose that the single mode fiber segments are modeled using
\eqref{eq:NLSFilter} with $\epsilon >0$. 
Then the essential spectrum of the monodromy operator,
$\mathcal M$, in \eqref{eq:Linearization of the round trip operator}
is given by
\begin{equation}\label{eq:EssSpecsEqual}
\sigma_{\rm{ess}}(\mathcal M) \,\,=\,\, \sigma_{\rm{ess}}(\mathcal M_\infty).
\end{equation}
\end{theorem} 

\begin{rem}
For simplicity we state and prove this theorem for the lumped  model of the stretched pulse laser discussed in Section~\ref{sec:Mathematical Model}.
However, \eqref{eq:EssSpecsEqual} also holds for a wide range of lumped models of fiber lasers. Specifically, as we will see in the proof, in addition to the hypotheses made about the fiber segments, we just require that the linearizations, $\mathcal U$ and $\mathcal U_\infty$, of the transfer operators of the input-output devices in the model satisfy
\begin{equation}\label{eq:DiscOpsBounded}
\mathcal U, \mathcal U_\infty \in \mathcal B( L^2(\mathbb R, \mathbb C^2) ) \cap
\mathcal B( H^2(\mathbb R, \mathbb C^2) ),
\end{equation}
and that an analogue of Theorem~\ref{thm:CompactPerturbationFSA} below holds for each of them.
\end{rem}

\StartRed

\section{Simulation Results}\label{sec:SimResults}

In this section we use  formula \eqref{eq:Formula for eigenvalues in essential spectrum} for the essential spectrum  to provide
some insights into the roles that the saturable absorber and the saturation of the gain in the fiber amplifier play in stabilizing the
periodically  stationary pulse circulating in the laser. 
Further details  can be found in~\cite{shinglot2022continuous}.

Although we are not modeling it here, in addition to its role in pulse amplification,
the fiber amplifier adds spontaneous emission noise  to the system~\cite{giles1991modeling}, which---among other effects such as random
timing and phase shifts of the pulse---manifests itself as a random superposition 
of continuous wave perturbations. 
If the essential spectrum, $\sigma_{\operatorname{ess}}(\mathcal M)$, lies inside the unit disc in $\mathbb C$, then these continuous wave perturbations decay, which helps ensure pulse  stability.

From \eqref{eq:Saturable gain} we see that the  gain in the fiber amplifier simply depends on the pulse energy. Consequently, each round trip
the noise entering the fiber amplifier experiences the same gain as does the pulse.  
Furthermore, as the pulse propagates through the fiber amplifier,
spontaneous emmission noise that  is proportional to the 
gain is added to the system. The saturation of the gain  therefore plays a critical role in stabilizing the system, since the gain decreases as the pulse energy increases.

On the other hand, with the model we use for the saturable absorber 
the response is instantaneous, and is given by
\begin{equation}
    \psi_\text{out}(x) = \left( 1 - \frac{\ell_0}{1+|\psi_\text{in}(x)|^2/ P_{\text{sat}}}\right)\psi_\text{in}(x),\\
    \label{eq:SA_JZ}
\end{equation}
so that the value of the output at $x$ only depends on the input  at that $x$.
Therefore,  far from the pulse, where $\psi_\text{in} \approx 0$,
the loss is $\ell_0$, whereas in the center of the pulse the loss saturates and is less than $\ell_0$. 
Because the loss saturates at high power, the system can operate so that 
the gain in the fiber amplifier and the loss
in the saturable absorber balance for the pulse, while simultaneously  loss exceeds gain  far from the pulse. Consequently, noise far from the pulse can be suppressed relative to the peak power of the pulse. 
The larger $\ell_0$ is and/or the smaller $P_{\text{sat}}$ is in \eqref{eq:SA_JZ}, the more the saturable absorber suppresses noise far from the pulse, and the more stable the pulse is to
noise perturbations. 
Already in the 1975, Haus~\cite{haus1975theory} identified the need for a saturable
absorber to suppress the growth of continuous waves,
while balancing gain and loss for the pulse. 
Formula~\eqref{eq:Formula for eigenvalues in essential spectrum} for the essential spectrum of the monodromy operator quantifies this effect for the first time in a lumped model of a fiber laser.

To ensure that a continuous wave perturbation with frequency $\omega$
does not grow, we require that  
$|\lambda_\pm(\omega)| \leq 1$,
which, because of the Gaussian factor 
in~\eqref{eq:Formula for eigenvalues in essential spectrum}, 
 holds for all $\omega$ provided that
\begin{equation}\label{eq:GvL}
(\ell_{\operatorname{OC}})^2( 1 - \ell_0)^2G_{\text{Tot}}^{\text{FA}}
\leq 1,
\qquad\text{where  } 
G_{\text{Tot}}^{\text{FA}} = 
\exp\left\{ \int_0^{L_{\text{FA}}} g(\psi(t))\, dt\right\},
\end{equation}
is approximately equal to the energy gain in the fiber amplifier. 
That is, far from the pulse the loss experienced by continuous waves
must exceed the gain.
Although \eqref{eq:GvL} looks very simple,
 the essential spectrum  can depend in a complex way on the interplay between all the system parameters, since they all influence the shape of the pulse and hence the total gain, $G_{\text{Tot}}^{\text{FA}}$, in the fiber amplifier.

\begin{figure}[htb]
    \centering
    \includegraphics[width=0.335\linewidth]{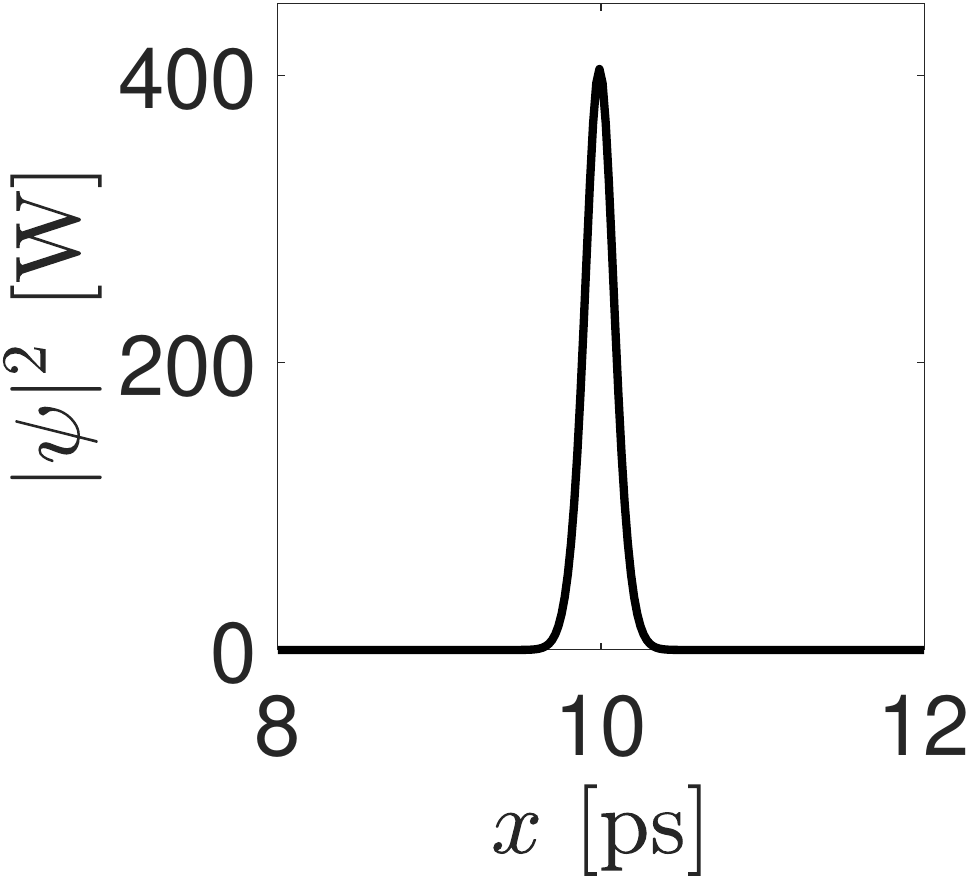}
    \includegraphics[width=0.29\linewidth]{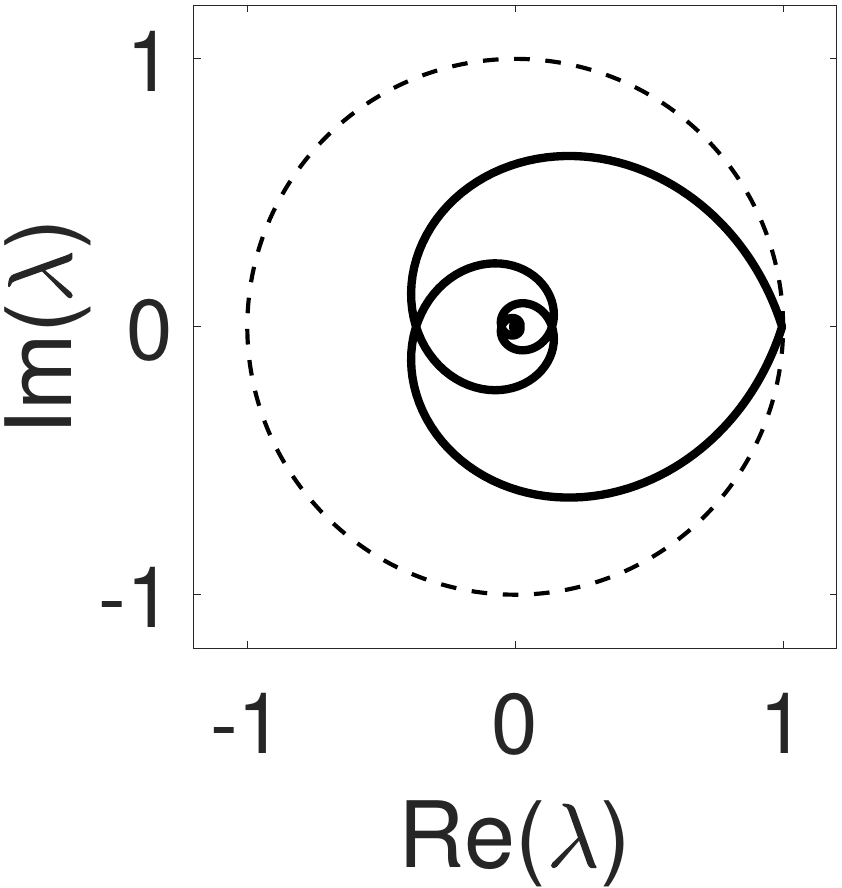}
    \includegraphics[width=0.31\linewidth]{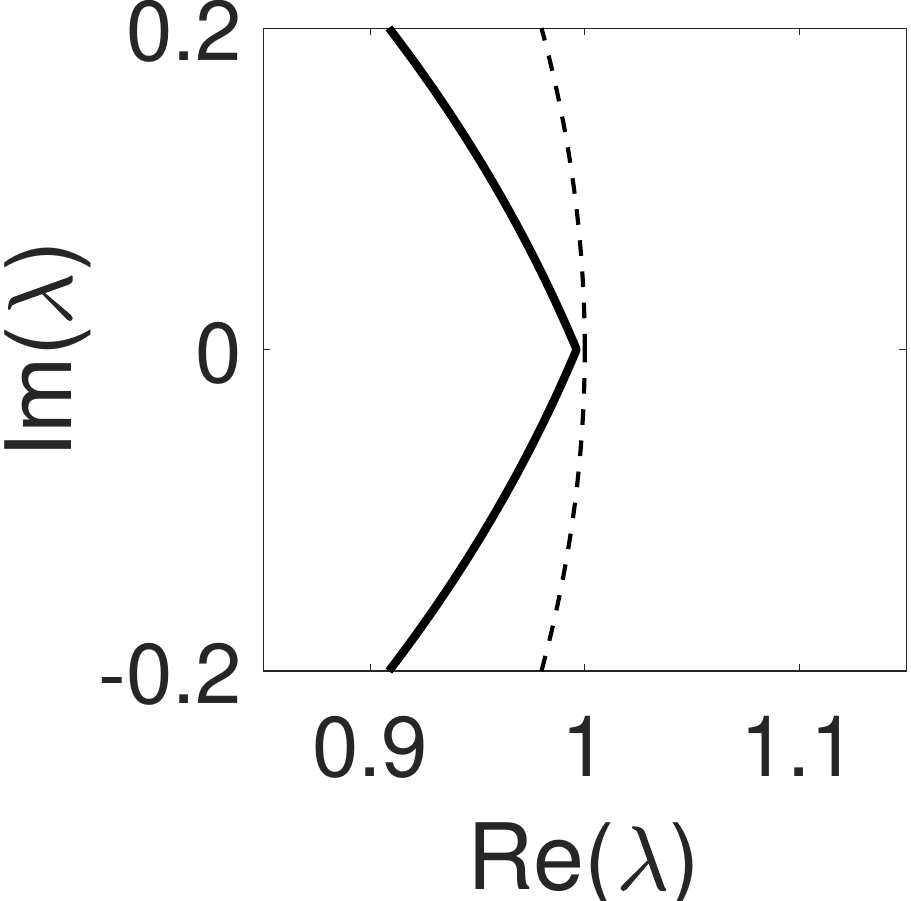}\\
       \includegraphics[width=0.335\linewidth]{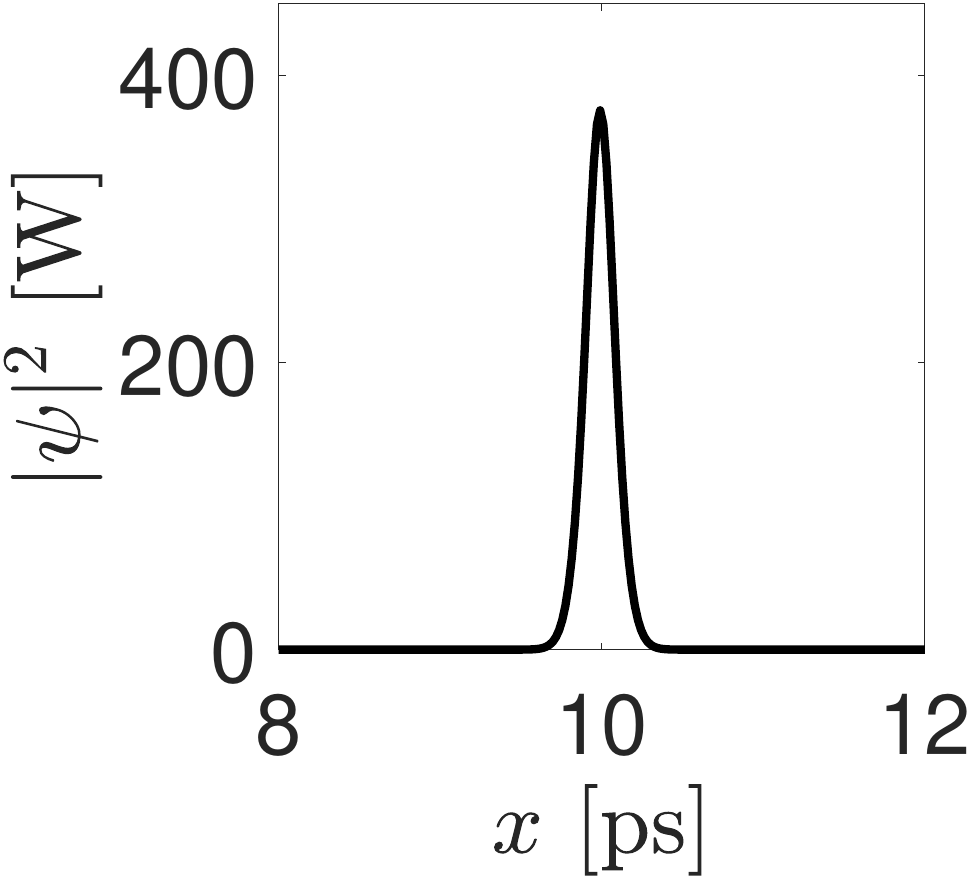}
    \includegraphics[width=0.29\linewidth]{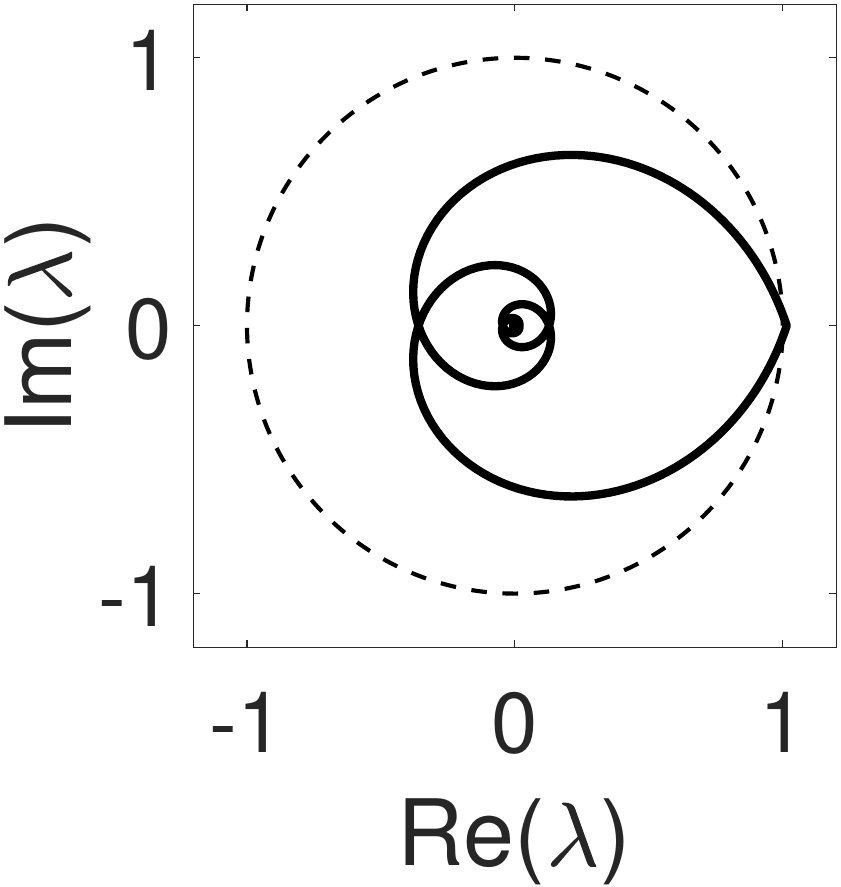}
    \includegraphics[width=0.31\linewidth]{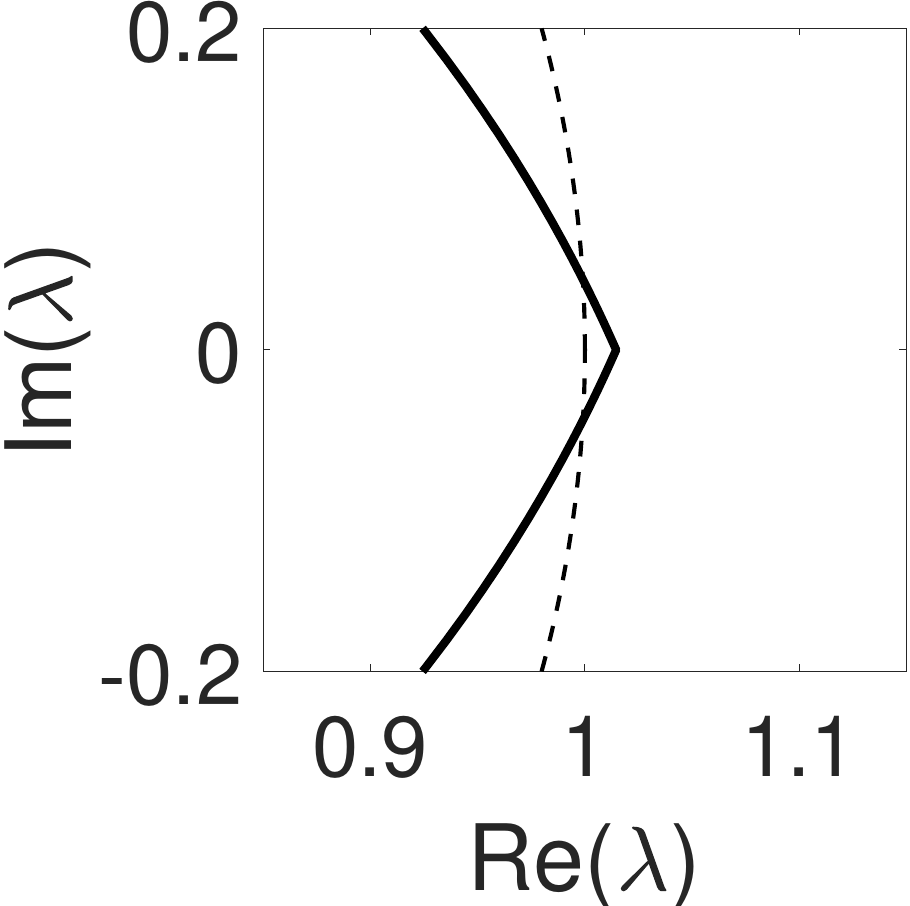}
    \caption{\textbf{Top row:} {\bf Left:} Periodically stationary pulse for $P_{\operatorname{sat}} = 200$~W. {\bf Center and right:} Essential spectrum, $\sigma_{\operatorname{ess}}(\mathcal{M})$, of the monodromy operator associated 
    with the pulse on the left.    \textbf{Bottom row:} Corresponding results for $P_{\text{sat}} = 1000$~W. In both cases, $\ell_0=0.05$.
    }
    \label{fig:PSPulsesEssSpec}
\end{figure}

\begin{figure}[htb]
    \centering
       \includegraphics[width=0.45\linewidth]{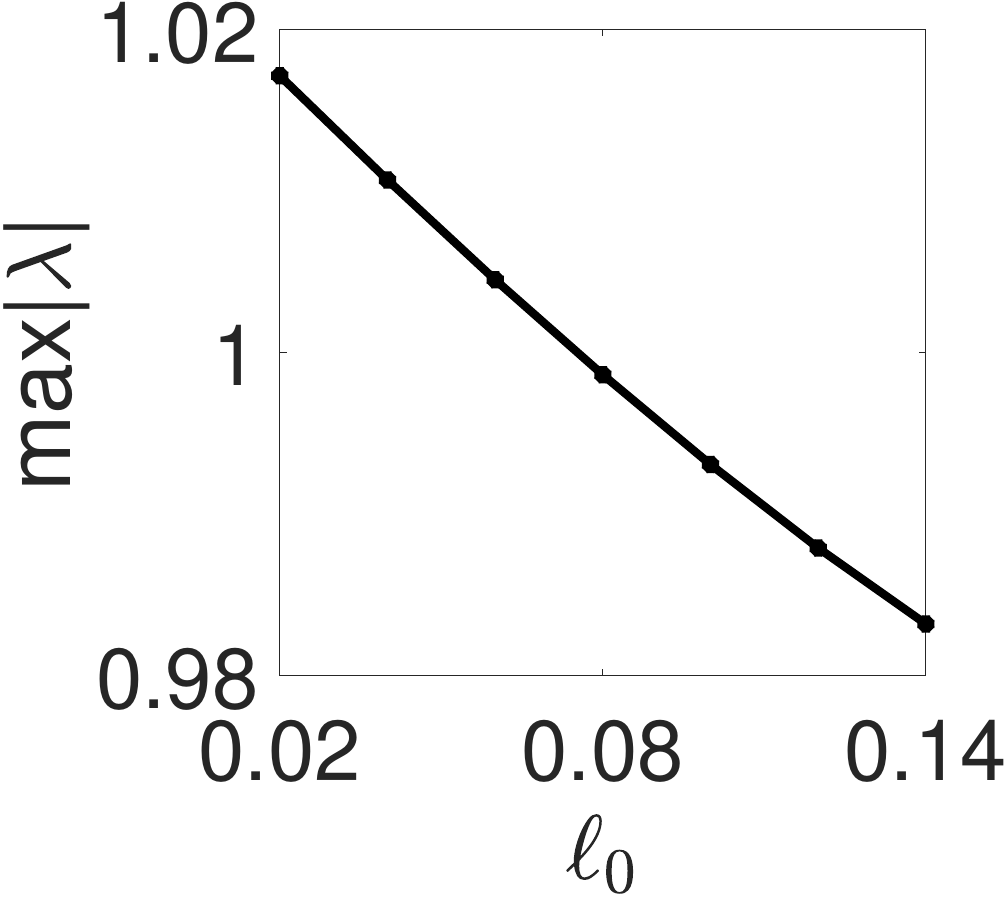}
          \hskip20pt
       \includegraphics[width=0.45\linewidth]{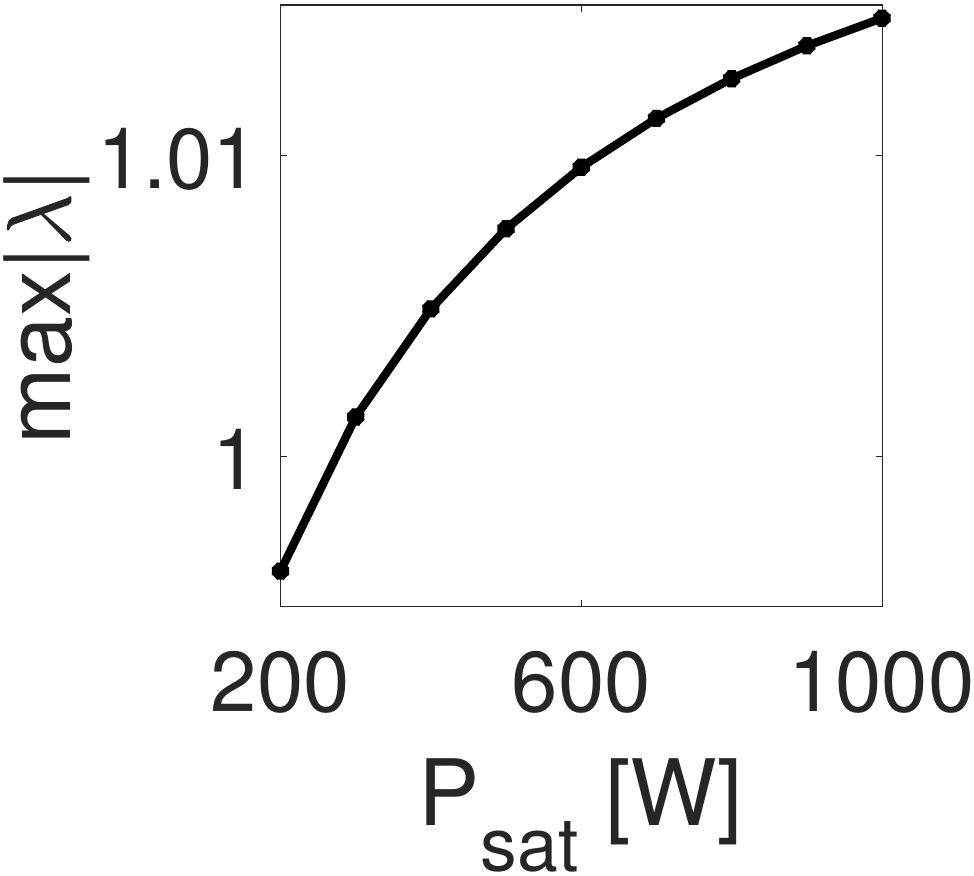}
        \caption{\textbf{Left:} A plot of the maximum real eigenvalue, $\operatorname{max} |\lambda|$, vs. $\ell_0$ when $P_{\text{sat}}=500$~W. \textbf{Right:}
        Corresponding plot in which   $P_{\text{sat}}$ is varied when $\ell_0=0.05$.
     }
         \label{fig:MaxLambda}
 \end{figure}   
 
For the simulation results we present here, we chose the parameters 
in the model to be similar to those in the experimental stretched pulse laser of Kim~\cite{kim2014sub}. The parameters for the saturable absorber are given 
below. The saturable absorber is followed by a segment of single mode fiber, SMF1, 
modeled by \eqref{eq:SMF_JZ}, with $\gamma=2\times10^{-3}$~(Wm)$^{-1}$,
$\beta_{\operatorname{SMF1}}=10$~kfs$^2$/m, (1~kfs$^2$ = $10^{-27}$~s$^2$), and $L_{\operatorname{SMF1}}=0.32$~m,
a fiber amplifier,  modeled by \eqref{eq:fiber amplifier},  with 
$g_0=6$m$^{-1}$, $E_{\text{sat}}=200$~pJ, $\Omega_g=50$~THz, $\gamma=4.4\times10^{-3}$~(Wm)$^{-1}$,
$\beta_{\operatorname{FA}}=25$~kfs$^2$/m, and $L_{\operatorname{FA}}=0.22$~m,
a second segment of single mode fiber, SMF2, with the same parameters as SMF1, but with $L_{\operatorname{SMF2}}=0.11$~m, a dispersion compensation element with $\beta_{\operatorname{DCF}}=-1~$kfs$^2$,
and a 50\% output coupler, modeled by \eqref{eq:Output coupler} 
with $\ell_{\operatorname{OC}}=\sqrt{0.5}$.

In the top row of Fig.~\ref{fig:PSPulsesEssSpec}, we show the results of simulations
performed when $P_{\text{sat}}=200$~W and $\ell_0=0.05$. 
The pulse, $\psi_0$, in the left panel was obtained by numerically minimizing the 
$L^2$-error between $\mathcal R (\psi_0)$ and $e^{i\theta} \psi_0$, over all
possible choices of $\theta$~\cite{shinglot2022continuous}.
In the center panel we plot the essential spectrum for the pulse
in the left panel. 
We observe that $\sigma_{\operatorname{ess}}(\mathcal M)$ consists of a pair of
counter-rotating spirals whose amplitudes rapidly decay to zero.
Since the peak power of the pulse entering the saturable absorber is comparable
to $P_{\text{sat}}$, the saturation of the loss is significant, which helps to stabilize
the pulse.
In the bottom row of Fig.~\ref{fig:PSPulsesEssSpec}, we show the corresponding results with $P_{\text{sat}}=1000$~W. In this case the saturation of the loss is
much weaker, and as we see in the far right panel, there is a range of 
low frequencies, $\omega$, for which $ |\lambda\pm( \omega)| > 1$ and
 continuous wave perturbations grow.

 In the left panel of Fig.~\ref{fig:MaxLambda}, we plot the largest value of 
 $| \lambda |$ as a function of  $\ell_0$ when $P_{\text{sat}}=500$~W. Since this value 
remains outside the unit circle as $\ell_0$ increases from 0.02 to 0.06, the pulse
is unstable over this range. It is only once the unsaturated gain is sufficiently large
that condition~\eqref{eq:GvL}
 holds and the essential spectrum is stable. 
Similarly, in the right panel, we show the largest value of 
 $| \lambda |$  as a function of $P_{\text{sat}}$ when $\ell_0=0.05$. 
Here, the pulse is unstable for $P_{\text{sat}}>300$~W, 
since then the saturation effect is too weak to ensure that the loss
experienced by the noise is sufficiently greater than that experienced by the 
pulse.

\EndRed

\section{Existence of the monodromy operator}\label{sec:ExistenceMonodromy}
To prove Theorem~\ref{thm:Properties of monodromy operator} we use the fact that the monodromy operator, $\mathcal{M}$, is the composition of the linearized transfer functions, $\mathcal{U}$, of each component of the laser. Therefore, we just need to establish the result for each of the operators, $\mathcal{U}$. 
For the single mode fiber segments and the dispersion compensation element, the result is a special case of 
 the corresponding result for the CQ-CGL equation 
given in Zweck et al. \cite[Theorem 4.1]{zweck2021essential}.
For the fast saturable absorber  and the fiber amplifier, the results are given in
Proposition~\ref{prop:Existence of evolution family for SA} and  Theorem~\ref{thm:Existence of evolution family of linearized FA}  below.

If $X$ is a Banach space, we let $\norm{\, \cdot \,}_X$ denote the norm on $X$. When the context is clear, we sometimes omit the subscript $X$ and simply write $\norm{\, \cdot \,}$.

\begin{proposition}
\label{prop:Existence of evolution family for SA}
Suppose that Hypothesis~\ref{hyp:Conditions on the solution of SA} holds. Then the transfer function, $\mathcal{U}^{\operatorname{SA}}$, given by \eqref{eq:SA linearized transfer function}
satisfies  the first two conclusions of Theorem~\ref{thm:Properties of monodromy operator}. \end{proposition}

\begin{proof}
To establish the first conclusion, we use the Cauchy-Schwarz inequality and the fact that $\ell(\boldsymbol{\psi}_{\text{in}}) \leq \ell_0$ (see \eqref{eq:SA linearized transfer function}) to obtain
\begin{equation}
\begin{aligned}
    \norm{\boldsymbol{u}_{\text{out}}}_{L^2(\mathbb{R}, \mathbb{C}^2)} &\leq  (1+ \ell(\boldsymbol{\psi}_{\text{in}})) \norm{\boldsymbol{u}_{\text{in}}}_{L^2(\mathbb{R}, \mathbb{C}^2)} + \frac{2 \ell^2(\boldsymbol{\psi}_{\text{in}})}{\ell_0 P_{\text{sat}}} |\boldsymbol{\psi}_{\text{in}}^T\boldsymbol{u}_{\text{in}}| \, \norm{\boldsymbol{\psi}_{\text{in}}}_{L^2(\mathbb{R}, \mathbb{C}^2)}\\
    &\leq \left( 1 + \ell_0 + \frac{2\ell_0}{P_{\text{sat}}} \norm{\boldsymbol{\psi}_{\text{in}}}_{L^2(\mathbb{R}, \mathbb{C}^2)}^2 \right) \norm{\boldsymbol{u}_{\text{in}}}_{L^2(\mathbb{R}, \mathbb{C}^2)}.
\end{aligned}
\label{eq:inequality for L2 norm of uout for SA}
\end{equation}
By Hypothesis~\ref{hyp:Conditions on the solution of SA},  $\boldsymbol{\psi}_{\text{in}} \in L^2(\mathbb{R}, \mathbb{C}^2)$. Therefore,  $\mathcal{U}^{\text{SA}} \in \mathcal{B}(L^2(\mathbb{R}, \mathbb{C}^2))$.
Similarly, to establish the second conclusion, we find that
\begin{equation}
    \norm{\boldsymbol{u}_{\text{out}}}_{H^2(\mathbb{R}, \mathbb{C}^2)} \leq \left( 1 + \ell_0 + \frac{2\ell_0}{P_{\text{sat}}} \norm{\boldsymbol{\psi}_{\text{in}}}_{H^2(\mathbb{R}, \mathbb{C}^2)}^2 \right) \norm{\boldsymbol{u}_{\text{in}}}_{H^2(\mathbb{R}, \mathbb{C}^2)}.
\label{eq:inequality for H2 norm of uout for SA}
\end{equation}  
By Hypothesis~\ref{hyp:Conditions on the solution of SA},  $\boldsymbol{\psi}_{\text{in}} \in H^2(\mathbb{R}, \mathbb{C}^2)$.     Therefore,  $\mathcal{U}^{\text{SA}} \in \mathcal{B}(H^2(\mathbb{R}, \mathbb{C}^2))$.    
 \end{proof}

Next, we establish the existence of an evolution family for the
linearization \eqref{eq:FA linearized equation} of the Haus master equation~\eqref{eq:fiber amplifier}, which models propagation in a fiber amplifier 
\StartRed
of length $L_{\text{FA}}$. 
\EndRed
Let $t \in [0,L_{\text{FA}}]$ be local time within the fiber amplifier and let $s \in [0,L_{\text{FA}}]$. We study solutions, $\boldsymbol{u}:[s,L_{\text{FA}}]\rightarrow H^2(\mathbb{R},\mathbb{C}^2)$, of 
\begin{equation}
    \begin{gathered}
    \partial_t \boldsymbol{u} = \mathcal{L}_{\text{FA}}(t)\boldsymbol{u}, \hspace{1cm} \text{for } 
    \StartRed
    0\leq s < t \leq L_{\text{FA}},
    \EndRed
    \\
    \hspace{-3.6cm} \boldsymbol{u}(s) = \boldsymbol{v},
    \end{gathered}
    \label{eq:IVP for evolution family of FA}
\end{equation}
where $\boldsymbol{v} \in H^2(\mathbb{R},\mathbb{C}^2)$. Here,  $\mathcal{L}_{\text{FA}}(t)$ is the  family of operators on $L^2(\mathbb{R},\mathbb{C}^2)$ given by
reformulating  \eqref{eq:FA linearized equation} as
\begin{equation}
   \mathcal{L}_{\text{FA}}(t) = \textbf{B}(t)\partial_x^2 + \widetilde{\textbf{M}}(t),
       \label{eq:linearized operator of FA for evolution family A}
   \end{equation}
where, setting $g(t) := g(\psi(t))$,  
\begin{equation}
        \textbf{B}(t) = \frac{g(t)}{2\Omega_g^2}\textbf{I} - \frac{\beta}{2}\textbf{J}
        \quad\text{and}\quad
        \widetilde{\textbf{M}}(t)\boldsymbol{u} = \widetilde{\textbf{M}}_1(t)\boldsymbol{u} - \boldsymbol{\phi}(t)\langle \boldsymbol{\psi}(t), \boldsymbol{u} \rangle.
            \label{eq:linearized operator of FA for evolution family B}
\end{equation}
Here, $\langle \, \cdot \, , \, \cdot \, \rangle$ is the $L^2$-inner product on $L^2(\mathbb{R}, \mathbb{C}^2)$ and
 \begin{equation}       
        \widetilde{\textbf{M}}_1(t) = \frac{g(t)}{2}\textbf{I} + \gamma |\boldsymbol{\psi}|^2 \textbf{J} + 2\gamma \textbf{J} \boldsymbol{\psi} \boldsymbol{\psi}^T
            \quad\text{and}\quad
        \boldsymbol{\phi}(t) = \frac{g^2(t)}{g_0 E_{\text{sat}}} \left\{ \left( 1 + \frac{\partial_x^2}{\Omega_g^2} \right) \boldsymbol{\psi} \right\}.
            \label{eq:linearized operator of FA for evolution family C}
\end{equation}

\begin{definition}[{\cite[5.5.3]{pazy2012semigroups}}]
\label{d:Evolution system}
A two parameter family of bounded linear operators, $\mathcal{U}(t,s), 0 \leq s \leq t \leq T$, on $X$ is called an \emph{evolution system} if 
\begin{enumerate}
    \item[(i)] $\mathcal{U}(s,s)=\mathcal{I}$, and $\mathcal{U}(t,r)\circ\mathcal{U}(r,s) = \mathcal{U}(t,s)$ for $0 \leq s \leq r \leq t \leq T$, and
    \item[(ii)] $(t,s) \to \mathcal{U}(t,s)$ is strongly continuous for $0 \leq s \leq t \leq T$.
\end{enumerate}
\end{definition}

\begin{definition}
Let $\mathbf{A} = \mathbf{A}(t,x):[0,\infty)\times\mathbb{R} \to \mathbb{C}^{2 \times 2}$ be a bounded matrix-valued function. We define
\begin{equation}
    \norm{\mathbf{A}}_{\infty} = \sup_{(t,x)} \norm{\mathbf{A}(t,x)}_{\mathbb{C}^{2 \times 2}}.
\end{equation}
\end{definition}

\begin{theorem}
\label{thm:Existence of evolution family of linearized FA}
Assume that Hypothesis~\ref{hyp:conditions on the solution of FA} holds 
\StartRed
in the fiber amplifier.
\EndRed
Then there exists a unique evolution operator,
\StartRed
$\mathcal{U}^{\operatorname{FA}}(t,s) \in \mathcal{B}(L^2(\mathbb{R}, \mathbb{C}^2))$, for $0 \leq s \leq t \leq L_{\text{FA}}$, where  $L_{\operatorname{FA}}$ is the length of the fiber amplifier,
\EndRed
such that 
\begin{enumerate}
    \item $\norm{\mathcal{U}^{\operatorname{FA}}(t,s)}_{\mathcal{B}(L^2(\mathbb{R}, \mathbb{C}^2))} \leq \exp [\norm{\widetilde{\textbf{M}}}_{\infty} (t-s)]$,
    \item $\mathcal{U}^{\operatorname{FA}}(t,s)(H^2(\mathbb{R}, \mathbb{C}^2)) \subset H^2(\mathbb{R}, \mathbb{C}^2)$,
    \item For each $s$, $\mathcal{U}^{\operatorname{FA}}(\cdot, s)$ is strongly continuous in that for all $\boldsymbol{v} \in L^2(\mathbb{R}, \mathbb{C}^2)$, the mapping $t \mapsto \mathcal{U}^{\operatorname{FA}}(t,s)\boldsymbol{v}$ is continuous, and
    \item For each $\boldsymbol{v} \in H^2(\mathbb{R}, \mathbb{C}^2)$, the function $\boldsymbol{u}(t) = \mathcal{U}^{\operatorname{FA}}(t,s)\boldsymbol{v}$ is the unique solution of the initial value problem \eqref{eq:IVP for evolution family of FA} for which 
     \StartRed
    $\boldsymbol{u} \in C([s, L_{\operatorname{FA}}), H^2(\mathbb{R}, \mathbb{C}^2))$ and $\boldsymbol{u} \in C^1((s,L_{\operatorname{FA}}),L^2(\mathbb{R}, \mathbb{C}^2))$.
  \EndRed
\end{enumerate}
\end{theorem}

\begin{proof}
The result follows from \StartRed \cite[Theorems 5.2.3 and 5.4.8]{pazy2012semigroups}\EndRed.
\Cref{lem:B(t) infinitesimal generator,lem:M bounded,lem:F is C1} below guarantee that the assumptions of these theorems hold.
\end{proof}

\begin{lemma}
\label{lem:B(t) infinitesimal generator}
    The linear operator, $\textbf{B}(t)\partial_x^2:H^2(\mathbb{R},\mathbb{C}^2) \subset L^2(\mathbb{R},\mathbb{C}^2) \rightarrow L^2(\mathbb{R},\mathbb{C}^2)$ is closed with domain $H^2(\mathbb{R},\mathbb{C}^2)$. Furthermore, $(0,\infty) \subset \rho(\textbf{B}(t)\partial_x^2)$ and the resolvent operator satisfies
\begin{equation}
    \norm{\mathcal{R}(\lambda:\textbf{B}(t)\partial_x^2)}_{\mathcal{B}(L^2(\mathbb{R},\mathbb{C}^2))} \leq \frac{1}{\lambda}, \hspace{1cm} \text{for all } \lambda > 0.
    \label{eq:condition on resolvent operator}
\end{equation}
Consequently, $\textbf{B}(t)\partial_x^2$ is the infinitesimal generator of a $C_0$-semigroup on $L^2(\mathbb{R},\mathbb{C}^2)$.
\end{lemma}
\begin{proof}     
\Cref{eq:condition on resolvent operator} follows immediately from \cite[Lemma 4.1]{zweck2021essential}. The proof is completed by invoking the Hille-Yosida Theorem~\cite[1.3.1]{pazy2012semigroups}.
\end{proof}

\begin{lemma}
\label{lem:M bounded}
    Assume that Hypothesis~\ref{hyp:conditions on the solution of FA} is met. Then there exists $K>0$ such that 
    \StartRed
    for all $t\in [0,L_{\operatorname{FA}}]$ 
    \EndRed
    \begin{equation}
        \norm{\widetilde{\textbf{M}}(t)}_{\mathcal{B}(L^2(\mathbb{R},\mathbb{C}^2))} < K.
        \label{eq:M tilde is bounded linear operator}
    \end{equation}
\end{lemma}

\begin{proof}
    We have
    \begin{equation}
\norm{\widetilde{\textbf{M}}(t)\boldsymbol{u}}_{L^2(\mathbb{R},\mathbb{C}^2)} 
       \leq 
  \norm{\widetilde{\textbf{M}}_1}_{\infty} \norm{\boldsymbol{u}}_{L^2(\mathbb{R},\mathbb{C}^2)}          
        + \norm{\boldsymbol{\phi}(t)\langle \boldsymbol{\psi}(t), \boldsymbol{u} \rangle}_{L^2(\mathbb{R},\mathbb{C}^2)}.
        \label{eq:bounded operator 1}
    \end{equation}
 Let   $\| \mathbf A\|_F$ denote the Frobenius norm of a matrix $\mathbf A$. 
  We estimate the first term in  \eqref{eq:bounded operator 1} by
        \begin{align*}
        \norm{\widetilde{\textbf{M}}_1}_{\infty}^2 
                       \leq& \sup_{(t,x) \in [0,L_{\operatorname{FA}}]\times\mathbb{R}} \norm{\widetilde{\textbf{M}}_1(t,x)}_F^2\\
        =& \sup_{(t,x) \in [0,L_{\operatorname{FA}}]\times\mathbb{R}} \,\,\sum_{i,j = 1}^{2} \left|\frac{g(t)}{2}\textbf{I}_{ij} + \gamma |\boldsymbol{\psi}(t,x)|^2 \textbf{J}_{ij} + 2\gamma \left[\textbf{J} \boldsymbol{\psi}(t,x) \boldsymbol{\psi}^T(t,x)\right]_{ij}\right|^2\\
        \leq& \sup_{(t,x) \in [0,L_{\operatorname{FA}}]\times\mathbb{R}} \bigg\{ \frac{g^2(t)}{4} \sum_{i,j = 1}^{2} |I_{ij}|^2 + \gamma^2|\boldsymbol{\psi}(t,x)|^4 \sum_{i,j = 1}^{2} |J_{ij}|^2\\
        &+ 4\gamma^2 \sum_{i,j = 1}^{2} \left|\left[J \boldsymbol{\psi}(t,x) \boldsymbol{\psi}^T(t,x)\right]_{ij}\right|^2 + \gamma g(t) |\boldsymbol{\psi}(t,x)|^2 \sum_{i,j = 1}^{2} |I_{ij}| |J_{ij}|\\
        &+ 4\gamma^2|\boldsymbol{\psi}(t,x)|^2 \sum_{i,j = 1}^{2} |J_{ij}| \left|\left[J \boldsymbol{\psi}(t,x) \boldsymbol{\psi}^T(t,x)\right]_{ij}\right|  \\
        &+ 2\gamma g(t)\sum_{i,j = 1}^{2} |I_{ij}| \left|\left[J \boldsymbol{\psi}(t,x) \boldsymbol{\psi}^T(t,x)\right]_{ij}\right| \bigg\}\\
        =& \sup_{(t,x) \in [0,L_{\operatorname{FA}}]\times\mathbb{R}} \left\{ \frac{g^2(t)}{2} + 10\gamma^2|\boldsymbol{\psi}(t,x)|^4 + 4\gamma g(t) |\Re(\boldsymbol{\psi}(t,x)) \Im(\boldsymbol{\psi}(t,x))| \right\}\\
        \leq& \frac{g_0^2}{2} + 
        \sup_{(t,x) \in [0,L_{\operatorname{FA}}]\times\mathbb{R}} \left\{        
        10\gamma^2  |\boldsymbol{\psi}(t,x)|^4 + 4\gamma g_0 |\Re(\boldsymbol{\psi}(t,x)) \Im(\boldsymbol{\psi}(t,x))| \right\},
    \end{align*}
 which is finite by Hypothesis~\ref{hyp:conditions on the solution of FA}.

    As for the second term in \eqref{eq:bounded operator 1}, by the Cauchy-Schwarz inequality,
    \begin{align*}
        \norm{\boldsymbol{\phi}(t)\langle \boldsymbol{\psi}(t), \boldsymbol{u} \rangle}_{L^2(\mathbb{R},\mathbb{C}^2)} &\leq \norm{\boldsymbol{\phi}(t)}_{L^2(\mathbb{R},\mathbb{C}^2)} \norm{\boldsymbol{\psi}(t)}_{L^2(\mathbb{R},\mathbb{C}^2)} \norm{\boldsymbol{u}}_{L^2(\mathbb{R},\mathbb{C}^2)}\\
        &\leq \frac{g_0}{E_{\text{sat}}} \norm{ \boldsymbol{\psi}(t) + \frac{\partial_x^2 \boldsymbol{\psi}(t)}{\Omega_g^2}}_{L^2(\mathbb{R},\mathbb{C}^2)} \norm{\boldsymbol{\psi}(t)}_{L^2(\mathbb{R},\mathbb{C}^2)} \norm{\boldsymbol{u}}_{L^2(\mathbb{R},\mathbb{C}^2)}\\
         &\leq \max \left\{1,\frac{1}{\Omega_g^2} \right\} \frac{g_0}{E_{\text{sat}}} \norm{\boldsymbol{\psi}(t)}_{H^2(\mathbb{R},\mathbb{C}^2)}^2 \norm{\boldsymbol{u}}_{L^2(\mathbb{R},\mathbb{C}^2)}.
    \end{align*}
 The result now follows, since  $\norm{\boldsymbol{\psi}(t)}_{H^2(\mathbb{R},\mathbb{C}^2)} < \infty$   by Hypothesis~\ref{hyp:conditions on the solution of FA}.
\end{proof}

Combining \cite[Theorem 5.2.3]{pazy2012semigroups} and 
Lemmas~\ref{lem:B(t) infinitesimal generator} and \ref{lem:M bounded}, 
we conclude that\\ $\{\mathcal{L}_{\text{FA}}(t)\}_{t\in[0,L_{\operatorname{FA}}]}$ is a stable family of infinitesimal generators of $C_0$-semigroups on $L^2(\mathbb{R},\mathbb{C}^2)$. This is the first assumption in 
\StartRed
\cite[Theorem 5.4.8]{pazy2012semigroups}.
\EndRed
The following Lemma establishes the second assumption.

\begin{lemma}
\label{lem:F is C1}
Suppose that Hypothesis~\ref{hyp:conditions on the solution of FA} holds. Then for each $\boldsymbol{v} \in H^2(\mathbb{R},\mathbb{C}^2)$, we have that $F(\cdot) = \mathcal{L}_{\operatorname{FA}}(\cdot)\boldsymbol{v}:(0,L_{\operatorname{FA}}) \rightarrow L^2(\mathbb{R},\mathbb{C}^2)$ is $C^1$.
\end{lemma}

\begin{proof}
       We show that $F$ is differentiable with $F'(t) = \partial_t \mathcal{L}_{\text{FA}}(t)\boldsymbol{v}$. The proof that $F'$ is continuous is similar. By Hypothesis~\ref{hyp:conditions on the solution of FA}, $\mathcal{L}_{\text{FA}}(t)\boldsymbol{v}$, $\partial_t \mathcal{L}_{\text{FA}}(t)\boldsymbol{v} \in L^2(\mathbb{R},\mathbb{C}^2)$. 
 In Appendix~\ref{AppendixC1}, we show that
    \begin{equation}
        \begin{aligned}
        &\norm{F(t+h) - F(t) -hF'(t)}_{L^2(\mathbb{R},\mathbb{C}^2)}\\
        \leq\; & \bigg\{ 2\sqrt{2}h G_1(h) + 2\sqrt{2}h G_2(h) + \frac{g_0 C}{E_{\text{sat}}} h \norm{\boldsymbol{\psi}(t+h)}_{H^2(\mathbb{R},\mathbb{C}^2)} G_3(h)\\
        &+ \frac{2g_0^2 C}{E_{\text{sat}}^2}  h^2 \sup_{\tau \in (t,t+h)} \abs{E^{\prime}(\tau)} \norm{\boldsymbol{\psi}(\tau)}_{H^2(\mathbb{R},\mathbb{C}^2)} \norm{\partial_t \boldsymbol{\psi}(t)}_{H^2(\mathbb{R},\mathbb{C}^2)}\\
        &+ \frac{2g_0 C}{E_{\text{sat}}}h^2 \sup_{\tau \in (t,t+h)} \norm{\partial_t \boldsymbol{\psi}(\tau)}_{H^2(\mathbb{R},\mathbb{C}^2)} \norm{\partial_t \boldsymbol{\psi}(t)}_{H^2(\mathbb{R},\mathbb{C}^2)}\\   &+ h G_4(h) \norm{\boldsymbol{\psi}(t)}_{H^2(\mathbb{R},\mathbb{C}^2)} \bigg\} \norm{\boldsymbol{v}}_{H^2(\mathbb{R},\mathbb{C}^2)},
        \end{aligned}
        \label{eq:Bound on F' new}
    \end{equation}
    where
    \begin{equation}
        \begin{aligned}
            &G_1(h) = \sup_{\tau \in (t,t+h)} \norm{(\partial_t \textbf{B})(\tau) - (\partial_t \textbf{B})(t)}_{\mathbb{C}^{2 \times 2}},\\
            &G_2(h) = \sup_{(\tau, x) \in (t,t+h)\times \mathbb{R}} \norm{(\partial_t \widetilde{\textbf{M}_1})(\tau, x) - (\partial_t \widetilde{\textbf{M}_1})(t, x)}_{\mathbb{C}^{2 \times 2}},\\
            &G_3(h) = \sup_{\tau \in (t,t+h)} \norm{(\partial_t \boldsymbol{\psi})(\tau) - (\partial_t \boldsymbol{\psi})(t)}_{H^2(\mathbb{R},\mathbb{C}^2)},\\
            &G_4(h) = \sup_{\tau \in (t,t+h)} \norm{(\partial_t \boldsymbol{\phi})(\tau) - (\partial_t \boldsymbol{\phi})(t)}_{L^2(\mathbb{R},\mathbb{C}^2)}.
        \end{aligned}
        \label{eq:Gestimates}
    \end{equation}
    Next, we observe that $\exists C>0$ such that
    \begin{equation}
        G_1(h) = C \sup_{\tau \in (t,t+h)} \abs{g^2(\tau) E^{\prime}(\tau) - g^2(t) E^{\prime}(t)}.
    \end{equation}
    By Hypothesis~\ref{hyp:conditions on the solution of FA} and the differentiation under the integral sign theorem \cite{jones2001lebesgue}, $g$ and $E^{\prime}$ are $C^1$ which implies that $G_1(h) \to 0$ as $h \to 0$. Also by Hypothesis~\ref{hyp:conditions on the solution of FA}, and applying the Lebesgue dominated convergence theorem as needed, we conclude that $G_j(h) \to 0$ as $h \to 0$ for $j = 2,3,4$. Consequently, 
    \begin{equation}
        \norm{F(t+h) - F(t) -hF^{\prime}(t)}_{L^2(\mathbb{R},\mathbb{C}^2)} \leq hG(h),
    \end{equation}
    where $\lim_{h\to 0} G(h) = 0$. Hence, $F$ is differentiable as required.
    \end{proof}
    
%---------------------------------------------------------------------------

\section{Spectrum of a Multiplication Operator on $L^2(\mathbb R, \mathbb C^2)$}
\label{sec:Formula for the Spectrum of a Multiplication Operator}
The essential spectrum of the asymptotic linearized operator, $\mathcal{M}_{\infty}$, is equal to the spectrum of its Fourier transform, $\widehat{\mathcal{M}}_{\infty}$, which is a multiplication operator on $L^2(\mathbb{R}, \mathbb{C}^2)$. 
In this section, we derive a formula for the spectrum of a general class of multiplication operators on $L^2(\mathbb{R}, \mathbb{C}^2)$. 
The proof  is based on that of a similar well-known formula for multiplication operators on $L^2(\mathbb{R}, \mathbb{C})$~\cite[Prop. 4.2]{engel2001one}. 

\begin{definition}
\label{d:Multiplication operator}
Let  $\textbf{Q} : \mathbb{R} \to \mathbb{C}^{2 \times 2}$. The \emph{multiplication operator}, $\mathcal{M}_{\textbf{Q}}$, induced on $L^2(\mathbb{R}, \mathbb{C}^2)$ by $\textbf{Q}$ is 
defined by
\begin{align}
(\mathcal{M}_{\textbf{Q}} \boldsymbol{w})(x) &:= \textbf{Q}(x)\boldsymbol{w}(x) \text{ for all } \boldsymbol{w} \text{ in the domain}\\
D(\mathcal M_{\textbf{Q}})  &= \{ \boldsymbol{w}  \in L^2(\mathbb{R}, \mathbb{C}^2) \,
:\, \textbf{Q} \boldsymbol{w}
\in L^2(\mathbb{R}, \mathbb{C}^2) \}.
\end{align}
\end{definition}

\begin{proposition}\label{prop:MQ is densely defined and closed}
If $\textbf{Q} \in L^\infty(\mathbb R, \mathbb C^{2\times 2})$, then  $\mathcal{M}_{\textbf{Q}}$ is everywhere defined, bounded and closed, with
\begin{equation}
\norm{\mathcal{M}_{\textbf{Q}}}_{\mathcal{B}(L^2(\mathbb{R}, \mathbb{C}^2)))} \leq \norm{\textbf{Q}}_{\infty},
 \end{equation}
where
\begin{equation}
    \norm{\textbf{Q}}_{\infty} := \sup_{x \in \mathbb{R}} \norm{\textbf{Q}(x)}_{\mathbb{C}^{2\times 2}}.
    \label{eq:inifinity norm}
\end{equation}
\end{proposition}

We now state the main result of this section.

\begin{theorem}
\label{thm:Formula for spectrum of MQ}
    Let $\textbf{Q} \in L^{\infty}(\mathbb{R}, \mathbb{C}^{2 \times 2}) \cap C^0(\mathbb{R}, \mathbb{C}^{2 \times 2})$. If $\norm{\textbf{Q}(x)}_{\mathbb{C}^{2 \times 2}} \to 0$ as $x \to \pm\infty$, then the spectrum of $\mathcal{M}_{\textbf{Q}}$ is 
 given by
    \begin{equation}
    \begin{aligned}
        \sigma(\mathcal{M}_{\textbf{Q}}) &= \{ \lambda \in \mathbb{C} \,:\, \exists x \in \mathbb{R} \text{ such that } \emph{det}(\lambda \textbf{I} - \textbf{Q}(x)) = 0 \} \cup \{0\}\\
        &= \{ \lambda \in \mathbb{C} \,:\, \exists x \in \mathbb{R} \text{ such that } \lambda \in \sigma(\textbf{Q}(x)) \} \cup \{0\}.
    \end{aligned}
    \end{equation}
\end{theorem}

The proof of Theorem~\ref{thm:Formula for spectrum of MQ} relies on several preliminary results.
First, Proposition~\ref{prop:MQ is densely defined and closed} can be improved upon as follows. 

\begin{proposition}
\label{prop:MQ is bounded iff Q is bounded}
Suppose that 
$\textbf{Q} \in C^0(\mathbb{R}, \mathbb{C}^{2 \times 2})$. Then,
the operator $\mathcal{M}_{\textbf{Q}}$ is bounded if and only if $\textbf{Q}$ is bounded. In this case, 
\begin{equation}
    \norm{\mathcal{M}_{\textbf{Q}}}_{\mathcal{B}(L^2(\mathbb{R}, \mathbb{C}^2)))} = \norm{\textbf{Q}}_{\infty}.
    \label{eq:norm of MQ is infinity norm of Q}
\end{equation}
\end{proposition}

The proof of this proposition relies on the following well-known result on the Dirac delta distribution.

\begin{lemma}
\label{prop:limit of integral of psi times g}
Let $g \in L^1(\mathbb{R})$ with $\int_{\mathbb{R}} g(x)dx = 1$. Set $g_{s,\delta}(x) = \frac{1}{\delta} g\left(\frac{x-s}{\delta}\right)$, where $\delta>0$. Then   $\lim_{\delta \to 0} \int_{\mathbb{R}} \phi(x) g_{s,\delta}(x) dx = \phi(s)$ for all $\phi \in L^{\infty}(\mathbb{R}) \cap C^0(\mathbb{R})$. That is,  for every $\epsilon>0$, there exists $\widetilde{\delta} = \widetilde{\delta}(\epsilon, \phi)$ such that 
\begin{equation}
    \phi(s) - \epsilon \leq \int_{\mathbb{R}} \phi(x) g_{s,\delta}(x) dx \leq \phi(s) + \epsilon, \quad \text{whenever } \delta \leq \widetilde{\delta}.
    \label{eq:limit of integral of psi times g}
\end{equation}
\end{lemma}

\begin{proof}[Proof of Proposition~\ref{prop:MQ is bounded iff Q is bounded}]
    If $\textbf{Q}$ is bounded, then $\mathcal{M}_{\textbf{Q}}$ is bounded by Proposition~\ref{prop:MQ is densely defined and closed}. Conversely, suppose $\mathcal{M}_{\textbf{Q}}$ is bounded. Then,
        \begin{equation}
        \norm{\mathcal{M}_{\textbf{Q}}}_{\mathcal{B}(L^2(\mathbb{R},\mathbb{C}^2))} \geq \norm{\mathcal{M}_{\textbf{Q}} \boldsymbol{w}}_{L^2(\mathbb{R},\mathbb{C}^2)},
    \end{equation}
    for all $\boldsymbol{w} \in L^2(\mathbb{R},\mathbb{C}^2)$ with $\norm{\boldsymbol{w}}_{L^2(\mathbb{R},\mathbb{C}^2)} = 1$.
    Fix $s \in \mathbb{R}$ and choose $\boldsymbol{w}(x) = \boldsymbol{w}_{s,\delta}(x) = \sqrt{g_{s,\delta}(x)}\boldsymbol{v}(x)$, for some vector $\boldsymbol{v}(x) \in \mathbb{C}^2$ and where $g_{s, \delta}$ is as in Proposition~\ref{prop:limit of integral of psi times g}.
    If we require that $\norm{\boldsymbol{v}(x)}_{\mathbb{C}^2} = 1$ for all $x$, then $\norm{\boldsymbol{w}}_{L^2(\mathbb{R}, \mathbb{C}^2)} = 1$ holds. Furthermore, 
     for each $x$, we can chose $\boldsymbol{v}(x)$ so that
    \begin{equation}
        \norm{\textbf{Q}(x)\boldsymbol{v}(x)}_{\mathbb{C}^2} = \norm{\textbf{Q}(x)}_{\mathbb{C}^{2 \times 2}}. 
    \end{equation}
    Then
    \begin{equation*}
        \norm{\mathcal{M}_{\textbf{Q}}}_{\mathcal{B}(L^2(\mathbb{R},\mathbb{C}^2))}^2 \geq \norm{\mathcal{M}_{\textbf{Q}} \boldsymbol{w}_{s,\delta}}_{L^2(\mathbb{R},\mathbb{C}^2)}^2  = \int_{\mathbb{R}} \norm{\textbf{Q}(x)}_{\mathbb{C}^{2\times 2}}^2 \; g_{s,\delta}(x) dx.
    \end{equation*}
    Let $\epsilon > 0$. Choosing $\phi(x) = \norm{\textbf{Q}(x)}_{\mathbb{C}^{2\times 2}}^2$ in Proposition~\ref{prop:limit of integral of psi times g} we find that there exists $\widetilde{\delta} = \widetilde{\delta}(\epsilon, s) > 0$ so that for all $\delta < \widetilde{\delta}$ 
    \begin{equation}
   \norm{\mathcal{M}_{\textbf{Q}}}_{\mathcal{B}(L^2(\mathbb{R},\mathbb{C}^2))}^2 \geq      \int_{\mathbb{R}} \norm{\textbf{Q}(x)}_{\mathbb{C}^{2\times 2}}^2 g_{s, \delta} (x) dx > \norm{\textbf{Q}(s)}_{\mathbb{C}^{2\times 2}}^2 - \epsilon.
    \end{equation}
     Therefore, 
    \begin{equation}
        \norm{\textbf{Q}}_{\infty} = \sup_{s \in \mathbb{R}} \norm{\textbf{Q}(s)}_{\mathbb{C}^{2\times 2}} \leq \norm{\mathcal{M}_{\textbf{Q}}}_{\mathcal{B}(L^2(\mathbb{R},\mathbb{C}^2))},
        \label{eq:Bound on infinity norm of Q}
    \end{equation}
    and so $\textbf{Q}$ is bounded, and \eqref{eq:norm of MQ is infinity norm of Q} holds by Proposition~\ref{prop:MQ is densely defined and closed}.
\end{proof}

Next, in Proposition~\ref{prop:Inverse of Q is bounded} and Proposition~\ref{prop:relation between Qu and detQ} we state some properties of a matrix valued function, $\textbf{Q} \in L^{\infty}(\mathbb{R}, \mathbb{C}^{2\times 2})$, which are used in the proof of Proposition~\ref{prop:MQ has bounded inverse iff Q has bounded inverse} below.

\begin{proposition}
\label{prop:Inverse of Q is bounded}
Let $\textbf{Q}:\mathbb{R} \to \mathbb{C}^{2\times 2}$ be continuous with $\norm{\textbf{Q}}_{\infty}<\infty$ and suppose that $0\notin \overline{\Im (\operatorname{\det} \textbf{Q})}$. Then $\textbf{Q}^{-1}:\mathbb{R} \to \mathbb{C}^{2\times 2}$ is continuous and $\norm{\textbf{Q}^{-1}}_{\infty}<\infty$. 
\end{proposition}

\begin{proof}
    Since, $0\notin \overline{\Im (\operatorname{\det} \textbf{Q})}$, there exists $\epsilon > 0$ such that $|\operatorname{\det} \textbf{Q}(x)| > \epsilon$, for all $x\in \mathbb{R}$. So,
    \begin{equation*}
        \textbf{Q}^{-1}(x) = \frac{1}{\operatorname{\det} Q(x)} 
        \begin{bmatrix} Q_{22}(x) & -Q_{12}(x) \\ -Q_{21}(x) & Q_{11}(x) \end{bmatrix}
    \end{equation*}
    exists and is continuous. Furthermore,
    \begin{equation}
        \norm{\textbf{Q}^{-1}(x)}_{\mathbb{C}^{2\times 2}}^2 \leq \norm{\textbf{Q}^{-1}(x)}_{F}^2
        = \frac{\norm{\textbf{Q}(x)}_{F}^2}{|\operatorname{\det} \textbf{Q}(x)|^2}
        \leq \frac{4\norm{\textbf{Q}(x)}_{\mathbb{C}^{2\times 2}}^2}{|\operatorname{\det} \textbf{Q}(x)|^2}
        \leq \frac{4}{\epsilon^2} \norm{\textbf{Q}}_{\infty}^2. 
        \qed    
    \end{equation}
\end{proof}

\begin{proposition}
\label{prop:relation between Qu and detQ}
Let $\textbf{Q} \in \mathbb{C}^{2\times 2}$ be a matrix. Then there exists a vector $\boldsymbol{u} \in \mathbb{C}^2$ with $\norm{\boldsymbol{u}}_{\mathbb{C}^2}=1$ so that 
\begin{equation}
    \norm{\textbf{Q}\boldsymbol{u}}_{\mathbb{C}^2}^2 \leq |\operatorname{\det} \textbf{Q}|.
    \label{eq:relation between Qu and detQ}
\end{equation}
\end{proposition}

\begin{rem} Geometrically $\textbf{Q}$ changes areas by a factor of $\abs{\det\,\textbf{Q}}$. This result says there exists a direction $\boldsymbol{u}$ in which $\textbf{Q}$ changes lengths by at most $\sqrt{\abs{\det\,\textbf{Q}}}$.
\end{rem}

\begin{proof} The following self evident claims leads to the proof of \eqref{eq:relation between Qu and detQ}.

\begin{claim}
\label{chap:Result for QR decomposition}
Let $\textbf{Q}=\textbf{UR}$ be the QR decomposition of $\textbf{Q}$, where $\textbf{U}$ is unitary and $\textbf{R}$ is upper triangular. Suppose \eqref{eq:relation between Qu and detQ} holds for $\textbf{R}$, then it also holds for $\textbf{Q}$.
\end{claim}

\begin{claim}
\label{chap:Result for multiple of Q}
Suppose $\textbf{Q} = \alpha \widetilde{\textbf{Q}}$ for some $\alpha \in \mathbb{C}$ and that the \eqref{eq:relation between Qu and detQ} holds for $\widetilde{\textbf{Q}}$. Then \eqref{eq:relation between Qu and detQ} also holds for $\textbf{Q}$. 
\end{claim}

By Claim~\ref{chap:Result for QR decomposition} it suffices to establish \eqref{eq:relation between Qu and detQ} for $\textbf{R} = \begin{bmatrix} a & b\\0 & d \end{bmatrix}$.

{\bf Case I:} If $a = 0$, let $\boldsymbol{u} = (1,0)$. Then $\textbf{R}\boldsymbol{u} = (0,0)$. Hence, $\norm{\textbf{R}\boldsymbol{u}}_{\mathbb{C}^2}^2 = 0 = |\operatorname{\det} \textbf{R}|$, and so \eqref{eq:relation between Qu and detQ} holds.
 
 {\bf Case II:} If $a \neq 0$, then by Claim~\ref{chap:Result for multiple of Q} we just need to show that \eqref{eq:relation between Qu and detQ} holds for matrices $\widetilde{\textbf{R}}$ of the form $\widetilde{\textbf{R}} = \begin{bmatrix} 1 & b\\0 & d \end{bmatrix}$.
    If $|d|\geq 1$,  we choose $\boldsymbol{u}=(1,0)$ to obtain $\norm{\widetilde{\textbf{R}}\boldsymbol{u}}_{\mathbb{C}^2}^2 = 1 \leq |d| = |\operatorname{\det} \widetilde{\textbf{R}}|$.
    Finally, if $|d| < 1$,  choosing $\boldsymbol{u} = \left( -b/ \sqrt{1+|b|^2}, 1/ \sqrt{1+|b|^2} \right)$ we obtain $\widetilde{\textbf{R}}\boldsymbol{u} = (0,d)/\sqrt{1+|b|^2}$. Hence, $\norm{\widetilde{\textbf{R}}\boldsymbol{u}}_{\mathbb{C}^2}^2 = |d|^2/(1+|b|^2) \leq |d|^2 \leq |d| = |\operatorname{\det}\widetilde{\textbf{R}}|$.
\end{proof}

\begin{proposition}
\label{prop:MQ has bounded inverse iff Q has bounded inverse}
Let $\textbf{Q}: \mathbb{R} \to \mathbb{C}^{2 \times 2}$ be continuous with $\norm{\textbf{Q}}_{\infty} < \infty$. Then the operator $\mathcal{M}_{\textbf{Q}}$ has a bounded inverse if and only if  $0 \notin \overline{\text{Im}(\text{det}\textbf{Q})}$. In that case, $\textbf{Q}$ has a bounded inverse, $\textbf{Q}^{-1}$, and 
\begin{equation*}
    \mathcal{M}_{\textbf{Q}}^{-1} = M_{\textbf{Q}^{-1}}.
\end{equation*}
\end{proposition}

\begin{proof}
    Suppose $0 \notin \overline{\text{Im}(\text{det}\textbf{Q})}$. By Proposition~\ref{prop:Inverse of Q is bounded},  $\norm{\textbf{Q}^{-1}}_{\infty} \leq \infty$.
    Hence, by Proposition~\ref{prop:MQ is bounded iff Q is bounded}, $\mathcal{M}_{\textbf{Q}}^{-1}$ is bounded and 
    \begin{equation}
        \norm{\mathcal{M}_{\textbf{Q}}^{-1}}_{\mathcal{B}(L^2(\mathbb{R}, \mathbb{C}^2))} = \norm{\textbf{Q}^{-1}}_{\infty} \leq \infty.
    \end{equation}
    Conversely, suppose that $\mathcal{M}_{\textbf{Q}}$ has a bounded inverse. Then for all $\boldsymbol{w} \in L^2(\mathbb{R}, \mathbb{C}^2)$, 
    \begin{equation}
        \gamma:= \frac{1}{\norm{\mathcal{M}_{\textbf{Q}}^{-1}}_{\mathcal{B}(L^2(\mathbb{R}, \mathbb{C}^2))}} \leq \frac{\norm{\mathcal{M}_{\textbf{Q}} \boldsymbol{w}}_{L^2(\mathbb{R}, \mathbb{C}^2)}}{\norm{\boldsymbol{w}}_{L^2(\mathbb{R}, \mathbb{C}^2)}}.
        \label{eq:definition of gamma}
    \end{equation}
    We will show that for all $x \in \mathbb{R}$
    \begin{equation}
        \abs{\text{det} \textbf{Q}(x)} > \frac{\gamma^2}{8},
    \end{equation}
    and hence $0 \notin \overline{\text{Im}(\text{det}\textbf{Q})}$.
    
    Assume for the sake of contradiction that there exists $s \in \mathbb{R}$ such that 
    \begin{equation}
        |\text{det} \textbf{Q}(s)| \leq \frac{\gamma^2}{8}.
        \label{eq:bound on detQ}
    \end{equation}
    Let $\boldsymbol{w}(x) = \boldsymbol{w}_{s,\delta}(x) = \sqrt{g_{s,\delta}(x)}\boldsymbol{u}(x)$, where $g_{s,\delta}(x)$ is as in Proposition~\ref{prop:limit of integral of psi times g} and, using Proposition~\ref{prop:relation between Qu and detQ}, 
   for each $x \in \mathbb{R}$,  $\boldsymbol{u}(x) \in \mathbb{C}^2$ is chosen so that  $\norm{\boldsymbol{u}(x)}_{\mathbb{C}^2} = 1$ and
    \begin{equation}
        \norm{\textbf{Q}(x)\boldsymbol{u}(x)}_{\mathbb{C}^2}^2 \leq \left| \text{det}\textbf{Q}(x) \right|.
        \label{eq:norm of Qw and det Q}
    \end{equation}
    Let $\epsilon>0$. By \eqref{eq:norm of Qw and det Q} and Proposition~\ref{prop:limit of integral of psi times g} there exists $\delta > 0$ so that 
    \begin{align*}
        \norm{\mathcal{M}_{\textbf{Q}} \boldsymbol{w}_{s,\delta}}_{L^2(\mathbb{R}, \mathbb{C}^2)}^2 
        &= \int_{\mathbb{R}} \norm{\textbf{Q}(x)\sqrt{g_{s,\delta}(x)}\boldsymbol{u}(x)}_{\mathbb{C}^2}^2 dx \leq \int_{\mathbb{R}} g_{s,\delta}(x)  \left| \text{det}\textbf{Q}(x) \right| dx\\
        &< |\text{det}\textbf{Q}(x)| + \epsilon < \frac{\gamma^2}{8} + \epsilon.
    \end{align*}
    Choosing $\epsilon = \frac{\gamma^2}{8}$ and applying our assumption \eqref{eq:bound on detQ} we find that 
    \begin{equation}
        \norm{\mathcal{M}_{\textbf{Q}} \boldsymbol{w}_{s,\delta}}_{L^2(\mathbb{R}, \mathbb{C}^2)} \leq \frac{\gamma}{2}, 
    \end{equation}
    which is a contradiction to \eqref{eq:definition of gamma}. Therefore, for all $x \in \mathbb{R}$
  $
        |\text{det} \textbf{Q}(x)| > \frac{\gamma^2}{8}.
$
    Hence, $0 \notin \overline{\text{Im}(\text{det}\textbf{Q})}$. Finally, using \eqref{prop:Inverse of Q is bounded}, we conclude that $\norm{\textbf{Q}^{-1}}_{\infty} \leq \infty$.
\end{proof}

\begin{proof}[Proof of Theorem~\ref{thm:Formula for spectrum of MQ}]
    By Proposition~\ref{prop:MQ has bounded inverse iff Q has bounded inverse} 
        \begin{align*}
        \lambda \in \rho(\mathcal{M}_{\textbf{Q}}) &\iff M_{\lambda - Q} \text{ has a bounded inverse}\\ 
        &\iff 0 \notin \overline{\Im (\text{det} (\lambda \textbf{I} - \textbf{Q}))}\\
        &\iff \exists \epsilon > 0 \text{ such that } \forall x \in \mathbb{R} \quad |\text{det} (\lambda \textbf{I} - \textbf{Q}(x))| \geq \epsilon.
    \end{align*}
    Therefore,
    \begin{equation}
        \begin{aligned}
        \lambda \in \sigma (\mathcal{M}_{\textbf{Q}}) &\iff \lambda \notin \rho(\mathcal{M}_{\textbf{Q}})\\
        &\iff \forall \epsilon > 0 \, \exists x \in \mathbb{R} \text{ such that } |\text{det} (\lambda \textbf{I} - \textbf{Q}(x))| < \epsilon.
        \end{aligned}
        \label{eq:Condition on eigenvalue in spectrum of MQ}
    \end{equation}
    Let 
    \begin{equation}
        \widetilde{\sigma}(\mathcal{M}_{\textbf{Q}}) = \{ \lambda \in \mathbb{C} \,:\, \exists x \in \mathbb{R} \text{ such that } \text{det}(\lambda \textbf{I} - \textbf{Q}(x)) = 0 \}.
    \end{equation}
    Then $\widetilde{\sigma}(\mathcal{M}_{\textbf{Q}}) \subseteq \sigma(\mathcal{M}_{\textbf{Q}})$. Let $\lambda \in \sigma(\mathcal{M}_{\textbf{Q}})\backslash \widetilde{\sigma}(\mathcal{M}_{\textbf{Q}})$.
    To complete the proof, we must show $\lambda = 0$. Choosing $\epsilon = 1/n$ in \eqref{eq:Condition on eigenvalue in spectrum of MQ},
    \begin{equation}
        \exists x_n \in \mathbb{R} \text{ such that } \text{det}(\lambda \textbf{I} - \textbf{Q}(x_n)) \leq 1/n.
        \label{eq:bound on l*I-Q(xn)}
    \end{equation}
  Suppose that the sequence $\{x_n\}_{n=1}^{\infty}$ is bounded. Then there exists a convergent subsequence $x_{n_k} \to x_*$. Since, we are assuming that $\textbf{Q}$ is continuous, 
    \begin{equation}
        \text{det} (\lambda \textbf{I} - \textbf{Q}(x_*)) = \lim_{n \to \infty} \text{det} (\lambda \textbf{I} - \textbf{Q}(x_n)) = 0.
    \end{equation}
    Therefore, $\lambda \in \widetilde{\sigma}(\mathcal{M}_{\textbf{Q}})$, which is a contradiction. Hence, $x_n$ is not bounded and so
    \begin{equation}
        \exists x_n \to \infty \text{ such that } \norm{\textbf{Q}(x_n)}_{\mathbb{C}^{2 \times 2}} \to 0.
    \end{equation}
    Let   
$a_n=   \operatorname{det}(\lambda \textbf{I} - \textbf{Q}(x_n)) =
  \lambda^2 - \text{trace}(\textbf{Q}(x_n)) \lambda + \text{det}(\textbf{Q}(x_n))$.
    Therefore,
    \begin{equation}
        \lambda = \frac{1}{2} \left[ \text{trace}(\textbf{Q}(x_n)) \pm \sqrt{\text{trace}^2(\textbf{Q}(x_n)) - 4(\text{det}(\textbf{Q}(x_n)) -a_n) } \right].
    \end{equation}
    Now, by \eqref{eq:bound on l*I-Q(xn)}, $a_n \to 0$ and by assumption $\norm{\textbf{Q}(x_n)}_F \to 0$ as $n \to \infty$. Therefore, $\lambda = 0$ must hold.
\end{proof}

%---------------------------------------------------------------------------
\section{The Essential Spectrum of the Asymptotic Monodromy Operator}
\label{sec:The Essential Spectrum of the Asymptotic Monodromy Operator}

In this section we prove Theorem~\ref{thm:Theorem for formula of essential spectrum}
which gives the formula for the essential spectrum of $\mathcal M_\infty$.
The proof relies on the following two results.

\begin{lemma}
\label{lem:Formula for exponential of A}
Let $\mathbf{A}(a,b) = \begin{bmatrix} a&-b \\ b&a \end{bmatrix}$. Then
\begin{equation}
    e^{\mathbf{A}(a,b)} = e^a \textbf{R}(b),
    \label{eq:Formula for exponential of A}
\end{equation}
where $\textbf{R}(b) = \begin{bmatrix} \cos{b}&-\sin{b} \\ \sin{b}&\cos{b} \end{bmatrix}$ is a rotation matrix.
\end{lemma}
\begin{proof}
Diagonalize $\mathbf{A}(a,b)$ and use Euler's formula. 
\end{proof}

Next, working with Definition~\ref{d:Spectra}, we have the following result.

\begin{proposition}
\label{prop:essential spectrum of asymptotic monodromy operator is that of its FT}
Let $\mathcal{M}_{\infty}: L^2(\mathbb{R}, \mathbb{C}^2) \to L^2(\mathbb{R}, \mathbb{C}^2)$ be the asymptotic monodromy operator given by \eqref{eq:Asymptotic linearized roundtrip operator}. Then 
\begin{equation}
    \sigma_{\emph{ess}}(\mathcal{M}_{\infty}) = \sigma_{\emph{ess}}(\widehat{\mathcal{M}}_{\infty}),
\end{equation}
where
\begin{equation}
    \widehat{\mathcal{M}}_{\infty} = \mathcal{F} \circ \mathcal{M}_{\infty} \circ \mathcal{F}^{-1}.
    \label{eq:Fourier transform of asymptotic monodromy operator}
\end{equation}
Here, $\mathcal{F}: L^2(\mathbb{R}, \mathbb{C}^2) \to L^2(\mathbb{R}, \mathbb{C}^2)$ is the Fourier transform.
\end{proposition}

\begin{proof}[Proof of Theorem~\ref{thm:Theorem for formula of essential spectrum}]
By Proposition~\ref{prop:essential spectrum of asymptotic monodromy operator is that of its FT} 
it suffices to compute $\sigma_{\text{ess}}(\widehat{\mathcal{M}}_{\infty})$. First,
we show that 
\begin{equation}
    \widehat{\mathcal{M}}_{\infty} = \widehat{\mathcal{U}}_{\infty}^{\text{OC}} \circ \widehat{\mathcal{U}}_{\infty}^{\text{DCF}} \circ \widehat{\mathcal{U}}_{\infty}^{\text{SMF2}} \circ \widehat{\mathcal{U}}_{\infty}^{\text{FA}} \circ \widehat{\mathcal{U}}_{\infty}^{\text{SMF1}} \circ \widehat{\mathcal{U}}_{\infty}^{\text{SA}}   
\end{equation}
is a multiplication operator by showing that each transfer function $\widehat{\mathcal{U}}_{\infty}$ is a multiplication operator. Here, for each laser component the transfer function $\widehat{\mathcal{U}}_{\infty}$ is the Fourier transform of the asymptotic linearized transfer function, $\mathcal{U}_{\infty}$, given in 
Section~\ref{sec:Linearization of the Round Trip Operator}. We then use Theorem~\ref{thm:Formula for spectrum of MQ} to obtain $\sigma_{\text{ess}}(\widehat{\mathcal{M}}_{\infty})$.

For the saturable absorber, 
\begin{equation}
     (\widehat{\mathcal{U}}_{\infty}^{\text{SA}} \widehat{\boldsymbol{u}}_{\text{in}})(\omega) = 
     (1 - \ell_0)
 \widehat{\boldsymbol{u}}_{\text{in}}(\omega),
\end{equation}
\StartRed
and, as in the derivation of \eqref{eq:DCF_JZ}, for the dispersion compensation element,
 \begin{equation}
    (\widehat{\mathcal{U}}_{\infty}^{\text{DCF}} \widehat{\boldsymbol{u}}_{\text{in}})(\omega) = 
    \exp \left\{ \mathbf{A}\left( 0, \frac{\omega^2}{2} \beta_{\text{DCF}}\right) \right\}
    \widehat{\boldsymbol{u}}_{\text{in}}(\omega).
\end{equation}
For the two single mode fiber segments, a similar formula holds for   
each solution operator, $\widehat{\mathcal{U}}_{\infty}^{\text{SMF}}$,
but with $\beta_{\text{DCF}}$ replaced by $\beta_{\text{SMF}} \text{L}_{\text{SMF}}$.
\EndRed
For the fiber amplifier, 
\begin{equation}\label{eq:FTFiberAmp}
    (\widehat{\mathcal{U}}_{\infty}^{\text{FA}} \widehat{\boldsymbol{u}}_{\text{in}})(\omega) =     
     \exp \left\{ \mathbf{A}\left( \frac{1}{2}\left(1-\frac{\omega^2}{\Omega_g^2} \right)\int_0^{L_{\text{FA}}} g(t)dt,\; \frac{\omega^2}{2} \beta_{\text{FA}} \text{L}_{\text{FA}} \right) \right\}    
    \widehat{\boldsymbol{u}}_{\text{in}}(\omega).
\end{equation}
\StartRed
Finally, $\widehat{\mathcal{U}}_{\infty}^{\text{OC}} = \mathcal P^{\text{OC}}$, which is given by \eqref{eq:Output coupler}.  
\EndRed

Combining these formulae, applying Lemma~\ref{lem:Formula for exponential of A}, and using
the fact that $\textbf{R}(\theta_1) \circ \textbf{R}(\theta_2) = \textbf{R}(\theta_1 + \theta_2)$ we have
\begin{equation}
    (\widehat{\mathcal{M}}_{\infty} \widehat{\boldsymbol{u}}_{\text{in}})(\omega) = \widehat{\textbf{M}}_{\infty}(\omega) \widehat{\boldsymbol{u}}_{\text{in}}(\omega),
\end{equation}
where
\begin{equation}
    \widehat{\textbf{M}}_{\infty}(\omega) = \frac{(1 - \ell_0)}{\sqrt{2}} \exp\left\{ \frac 12 \left( 1 - \frac{\omega^2}{\Omega_g^2} \right)\int_0^{L_{\text{FA}}} g(t)dt \right\} \textbf{R} \left( \frac{\omega^2}{2} \beta_{\text{RT}} \right).
    \label{eq:Asymptotic monodromy operator for laser}
\end{equation}
Using Theorem~\ref{thm:Formula for spectrum of MQ} with $\textbf{Q} = \widehat{\textbf{M}}_{\infty}(\omega)$, we obtain
\begin{equation}
    \begin{gathered}
        \sigma(\mathcal{M}_{\infty}) = \{\, \lambda_{\pm}(\omega) \in \mathbb{C} \,\, | \,\, \omega \in \mathbb{R} \,\} \cup \{ 0 \},\\
        \lambda_{\pm}(\omega) = \frac{(1-\ell_0)}{\sqrt{2}} \exp\left\{ \frac 12 \left( 1 - \frac{\omega^2}{\Omega_g^2} \right)\int_0^{L_{\text{FA}}} g(t)dt \right\} \exp\left\{ \pm i\frac{\omega^2}2\beta_{\text{RT}} \right\}.
    \end{gathered}
    \label{eq:Formula for spectrum of asymptotic monodromy operator}
\end{equation}
Finally we show that $\sigma_{\text{pt}}(\mathcal{M}_{\infty}) = \phi$, from which it follows that $\sigma_{\text{ess}}(\mathcal{M}_{\infty}) = \sigma(\mathcal{M}_{\infty})$. 
For this we recall that the point spectrum of a multiplication operator such as 
$\widehat{\mathcal M}_\infty$ is given by~\cite{engel2001one}
\begin{equation}
    \sigma_{\text{pt}}(\widehat{\mathcal{M}}_{\infty}) = \left\{ \lambda \in \mathbb{C} \, : \, \mu\left\{ \omega \in \mathbb{R} \, : \, \text{det}[\widehat{\textbf{M}}_{\infty}(\omega) - \lambda] = 0 \right\} > 0 \right\},
    \label{eq:Point spectrum of multiplication operator}
\end{equation}
where $\mu$ denotes Lebesgue measure on $\mathbb{R}$. Therefore, to show that $\sigma_{\text{pt}}(\widehat{\mathcal{M}}_{\infty}) = \phi$, we must show for all $\lambda \in \mathbb{C}$ that the set
\begin{equation}
    S_{\lambda} = \{ \omega \in \mathbb{R} \, : \, \lambda_+(\omega) = \lambda \text{ or } \lambda_-(\omega) = \lambda \},
\end{equation}
has measure zero. We observe that $\lambda_{\pm}: \mathbb{R} \to \mathbb{C}$ generically parametrizes a pair of counter-rotating spirals. 
Invoking the assumptions of the theorem, since
 $\ell_0 \neq 1$, and either $\beta_{\text{RT}} \neq 0$ or $\Omega_g < \infty$ and $\int_{0}^{L_{\text{FA}}} g(t)dt \neq 0$,  the mappings $\lambda_{\pm}:\mathbb{R} \to \mathbb{C}$ are at most countable-to-one, which implies that $S_{\lambda}$ has measure zero for all $\lambda \in \mathbb{C}$. 
\end{proof}

%---------------------------------------------------------------------------

\section{Relative compactness for the linearized differential operators in the fiber amplifier}
\label{sec:Relative compactness for the linearized differential operators in the fiber amplifier}

In this section we show that
the linearized differential operator in the fiber amplifier, $\mathcal L(t)$, is a relatively compact perturbation of the asymptotic linearized differential operator, $\mathcal L_\infty(t)$, provided that the nonlinear pulse
satisfies some reasonable weak regularity and exponential decay assumptions.

By \eqref{eq:FA linearized equation},
the operators $\mathcal L(t)$ and $\mathcal L_\infty(t)$ are related by
\begin{equation}\label{eq:OpL}
 \mathcal{L}(t)\,\,=\,\,  \mathcal{L}_{\infty}(t) \,\,+\,\,\mathbf M(t),
\end{equation}
where 
\begin{equation}\label{eq:LinfDef}
 \mathcal{L}_\infty(t)\,\,=\,\,  {\mathbf B}\left(\frac{g(t)}{2\Omega_g^2},\frac{\beta}{2}\right)\partial_x^2 \,\,+
 \,\,\frac 12 g(t)\mathcal  I,
 \end{equation}
with $\mathbf B(a,b) = \begin{bmatrix}a &-b\\b&a\end{bmatrix}$, and where
${\textbf{M}}(t)$ is the matrix-valued multiplication operator 
    \begin{align}
             {\textbf{M}}(t,\cdot)\boldsymbol{u} &= {\textbf{M}}_1(t,\cdot)\boldsymbol{u} - \boldsymbol{\phi}(t,\cdot)\langle \boldsymbol{\psi}(t,\cdot), \boldsymbol{u} \rangle,
              \label{eq:linearized operator of FA for evolution family new}
              \\
        {\textbf{M}}_1(t,\cdot) &= \gamma |\boldsymbol{\psi}(t,\cdot)|^2 \textbf{J} + 2\gamma \textbf{J} \boldsymbol{\psi}(t,\cdot) \boldsymbol{\psi}^T(t,\cdot).
      \end{align}      
Here $\boldsymbol \psi$ is the pulse  about which the Haus master equation \eqref{eq:fiber amplifier} is linearized and $\boldsymbol\phi$ is given by   
 \eqref{eq:linearized operator of FA for evolution family C}. 
Note that here we have chosen $\mathbf M$ so that 
$\mathbf M(t,x) \to \mathbf 0$ as $x\to\pm\infty$.

\begin{theorem}
\label{thm:EssSpecRCP}
Assume that Hypothesis~\ref{hyp:conditions on the solution of FA} is met and that $(g_0/\Omega_g, \beta)\neq(0,0)$.
Then,  the differential operator,
$\mathcal L(t)$, given in \eqref{eq:OpL}, is a relatively compact perturbation of
 $\mathcal L_\infty(t)$ in that there exists a 
$\lambda\in \rho(\mathcal L_\infty)$ so that the operator
$ \mathbf M \circ (\mathcal L_\infty - \lambda)^{-1}$ on $L^2(\mathbb R, \mathbb C^2)$
 is compact.
\end{theorem}

 \begin{proof}
Using an idea of Kapitula, Kutz, and Sandstede~\cite{kapitula2004evans} in their paper on the Evans function for nonlocal equations, we observe that
\begin{equation}
\mathcal L = \mathcal L_\infty + \mathbf M_1 + \mathcal K \circ\mathcal J,
\end{equation}
where $\mathcal J: L^2(\mathbb R,\mathbb C^2) \to\mathbb C$ is given by
$\mathcal J(\boldsymbol{u}) \,\,=\,\, \langle \boldsymbol\psi(t,\cdot), \boldsymbol{u} \rangle,
$
and $\mathcal K : \mathbb C \to L^2(\mathbb R, \mathbb C^2)$ is given by
 $
\mathcal K(a) \,\,=\,\, a \boldsymbol\phi.
$
Under Hypothesis~\ref{hyp:conditions on the solution of FA}, the analogous result in Zweck et al.~\cite[Theorem~3.1]{zweck2021essential} guarantees that $\mathcal L_\infty + \mathbf M_1$ is a 
relatively compact perturbation of $\mathcal L_\infty$. 
The theorem now follows 
from the fact that $\mathcal K \circ \mathcal J$ is compact,
since it factors through the finite dimensional space, $\mathbb C$. 
 \end{proof}
 
%---------------------------------------------------------------------------
 
\section{Analyticity of asymptotic linearized operator in the fiber amplifier}\label{sec:analycity}
 
In this section, we show that the operator $\mathcal L_\infty(t)\, \mathcal U_\infty(t,s)$
is bounded on $L^2(\mathbb R,\mathbb C^2)$, where 
$\mathcal L_\infty(t)$ is the asymptotic linearized operator in the fiber amplifier 
given by \eqref{eq:LinfDef}, and
$\mathcal U_\infty(t,s)$ is the corresonding evolution family.
Zweck et al.~\cite{zweck2021essential} previously established an analogous result for
the constant-coefficient complex Ginzburg-Landau equation under the assumption
 that the spectral filtering coefficient in the equation is positive. 
These results  will be used in Section~\ref{sec:EssSpecEqual} to prove our main result, Theorem~\ref{thm:EssSpecEqual}.

We begin by recalling what it means for an operator to be sectorial~\cite{lunardi2012analytic, pazy2012semigroups}.

\begin{definition}
A linear operator $\mathcal A : D(\mathcal A) \subset X \to X$ is \emph{sectorial} if $\exists$
 $\omega \in\mathbb R$, $\theta\in (\pi/2,\pi]$, $M>0$ so that
\begin{enumerate}
\item $\rho(\mathcal A) \supset S_{\theta,\omega} := \{ \lambda\in \mathbb C \,|\, 
\lambda\neq \omega, \, | \arg(\lambda-\omega) | < \theta \} $, and
\item $\| \mathcal R(\lambda : \mathcal A)  \|
\leq \frac{M}{| \lambda-\omega | }$, for all $\lambda \in S_{\theta,\omega}$.
\end{enumerate}
\end{definition}

\begin{rem}
Lunardi~\cite[Chapter 2]{lunardi2012analytic} shows that if  $\mathcal A$ is a sectorial operator then a family of operators, $\mathcal T(t)=e^{t\mathcal A}$, for $t>0$, can be defined in terms of a Dunford contour integral so as to satisfy the semigroup properties
\begin{enumerate}
\item $\mathcal T(0) = \mathcal I$,
\item $\mathcal T(s+t) =  \mathcal T(s) \mathcal T(t)$, for all $T,s\geq 0$,
\end{enumerate}
and for which the mapping $t\mapsto e^{t\mathcal A} \,\,:\,\,\mathbb R^+ \to 
\mathcal B(X)$ is analytic. Furthermore,
\begin{equation}
\frac{d}{dt} \,e^{t\mathcal A} = \mathcal A \,e^{t\mathcal A}.
\end{equation}
Such a semigroup is called an \textit{analytic semigroup}.
\end{rem}

We consider 
solutions, $\boldsymbol{u}:[s,L_{\text{FA}}]\rightarrow H^2(\mathbb{R},\mathbb{C}^2)$, of the initial value problem
\begin{equation}\label{eq:AsymptoticIVP}
\begin{aligned}
    \partial_t \boldsymbol{u} &= \mathcal{L}_\infty(t)\boldsymbol{u},\,\,  \text{for } t>s,\\
    \boldsymbol{u}(s) &= \boldsymbol{v}, \qquad \,\,\text{for } \boldsymbol{v} \in H^2(\mathbb{R},\mathbb{C}^2).
\end{aligned}
 \end{equation}

\begin{theorem}\label{FiberAmpAnalyticSemigroupThm}
Suppose that $0< \Omega_g < \infty$, that $(g_0,\beta)\neq(0,0)$, and that
$\boldsymbol \psi$ is differentiable with respect to $t$.
Then, there exists a unique evolution system, $\mathcal U_\infty(t,s)$,
for \eqref{eq:AsymptoticIVP} with $0\leq s \leq t \leq L_{\operatorname{FA}}$ so that
\begin{enumerate}
\item $\exists$ $C$ so that for all $s, t$ we have $\| \mathcal U_\infty(t,s)\|_{\mathcal B(L^2(\mathbb R,\mathbb C^2))} \leq C$,
\item $\mathcal U_\infty(s,s) = \mathcal I$  and $\mathcal U_\infty(t,r) = \mathcal U_\infty(t,s) \circ \mathcal U_\infty(s,r) $
for all $0\leq r\leq s < t \leq L_{\operatorname{FA}}$,
\item $\mathcal U_\infty(t,s) \in \mathcal B(L^2(\mathbb R,\mathbb C^2), 
H^2(\mathbb R,\mathbb C^2))$,
\item The mapping $t\mapsto \mathcal U_\infty(t,s)$ is differentiable for $t\in (s,L_{\text{FA}}]$ with values in $\mathcal B(L^2(\mathbb R,\mathbb C^2))$, and
$\partial_t  \mathcal U_\infty(t,s) \,\,=\,\, \mathcal L_\infty(t)  \mathcal U_\infty(t,s)$,
i.e., the function $\boldsymbol{u}(t) = \mathcal U_\infty(t,s)\boldsymbol{v}$ solves \eqref{eq:AsymptoticIVP}, and
\item $\exists$ $C_1$ and $C_2$ so that  $ \forall \,0\leq s<t\leq L_{\operatorname{FA}}$,
\begin{equation}
\| \mathcal L_\infty(t) \mathcal U_\infty(t,s) \|_{\mathcal B(L^2(\mathbb R,\mathbb C^2)}
\,\,\leq\,\, C_1\frac{ G(t,s)}{t-s} + C_2\frac{g(t)}2,
\end{equation}
where $G(t,s) = \exp\left( \tfrac 12 \int_s^t g(\tau)\, d\tau\right)$. 
\end{enumerate}
\end{theorem} 
 
\begin{proof} 
We will  show that the first four conclusions of the theorem hold for the evolution operator,
$\mathcal V_\infty(t,s)$,  associated to  the differential operator,
$\mathbf B(t)\partial_x^2$, and that 
\begin{equation}
\| (\mathbf B(t)\partial_x^2) \mathcal V_\infty(t,s) \|_{\mathcal B(L^2(\mathbb R,\mathbb C^2)}
\,\,\leq\,\, \frac{C_1}{t-s}.
\end{equation}
Then, the theorem immediately follows for the original operators $\mathcal L_\infty(t) = \mathbf B(t)\partial_x^2
+ \frac 12 g(t)\mathcal I$ with
$\mathcal U_\infty(t,s) = G(t,s) \mathcal V_\infty(t,s)$.
Applying a result from Lunardi~\cite[Chap. 6]{lunardi2012analytic}, 
to establish the result for $\mathcal V_\infty(t,s)$
it suffices to show that
the operator $\mathcal A=\mathcal A(t):=\mathbf B(t)\partial_x^2 $ is sectorial 
and that  $t\mapsto \mathcal A(t) \in \operatorname{Lip}([0,L_{\text{FA}}],
\mathcal B(H^2(\mathbb R,\mathbb C^2),L^2(\mathbb R,\mathbb C^2)))$.

To show $\mathcal A$ is sectorial, we first observe  that $\mathcal A$ is closed and that $\exists$ $\omega \geq 0$ so that
$\forall \lambda > \omega$, $\lambda\in\rho(\mathcal A)$ and 
$\| \mathcal R(\lambda : \mathcal A)\| \leq \frac{1}{\lambda-\omega}$.
Therefore, by \cite[Cor 1.3.8]{pazy2012semigroups}, $\mathcal A$ is the infinitesimal generator of a $C_0$-semigroup for which $\|\mathcal T(t) \| \leq e^{\omega t}$.
By the proof of \cite[Lemma 5.2]{zweck2021essential}, for all $\sigma>0$, 
\begin{equation}\label{eq:ResolventEstimate}
\| \mathcal R(\sigma+i\tau : \mathcal A) \| \leq \frac{C}{|\tau|}.
\end{equation}
To show that this condition implies that $\mathcal A$ is sectorial we make use of
\cite[Thm 2.5.2]{pazy2012semigroups}. However, as stated, this theorem requires that
the semigroup $\mathcal T(t)$ is uniformly bounded and $0\in \rho(\mathcal A)$. 
Since neither of these conditions is guaranteed to hold, we proceed as follows. 
Fix $\epsilon > 0$, define $\mathcal T_\epsilon(t) := e^{-(\epsilon + \omega)t}\,\mathcal T(t)$, and let $\mathcal A_\epsilon = \frac{\partial T_\epsilon}{\partial t}(0)$. Then
$\| \mathcal T_\epsilon (t)\| < 1$ is uniformly bounded and $0\in \rho(\mathcal A_\epsilon)$.
Therefore, the assumptions of \cite[Thm 2.5.2]{pazy2012semigroups} hold for
$\mathcal A_\epsilon$. 
Furthermore, \eqref{eq:ResolventEstimate} holds for $\mathcal A_\epsilon$ since
$\mathcal R(\sigma+i\tau : \mathcal A_\epsilon) = \mathcal R(\sigma+\epsilon+\omega + i\tau : \mathcal A)$.  So by \cite[Thm 2.5.2]{pazy2012semigroups},
$\exists$ $0 < \delta < \frac \pi 2$, $M>0$ so that
\begin{enumerate}
\item $\rho(\mathcal A_\epsilon) \supset \Sigma = \{ \lambda\in \mathbb C \,:\, 
| \arg \lambda | < \frac \pi 2 + \delta \} \cup \{0\} $, and
\item $\| \mathcal R(\lambda : \mathcal A_\epsilon) \|
\leq \frac{M}{| \lambda | }$, for all $\lambda \in \Sigma \setminus \{0\}$.
\end{enumerate}
Translating these conclusions back into statements about $\mathcal A$ itself,
we obtain
\begin{equation}
\| \mathcal R(\lambda : \mathcal A)\|  = 
\| \mathcal R(\lambda - (\epsilon+\omega): \mathcal A_\epsilon)\|
\leq \frac{M}{|\lambda - (\epsilon+\omega)|},
\end{equation}
whenever $\lambda - (\epsilon+\omega) \in \Sigma \setminus \{0\}$, which holds
precisely when $\lambda \in S_{\frac \pi 2 + \delta, \epsilon + \omega}$. 
Therefore, the operators $\mathcal A=\mathcal A(t)$ are sectorial.

Finally,  the mapping  $t\mapsto \mathcal A(t)$ is Lipschitz, since $\exists$ $C$ so that
\begin{equation*}
\| \mathcal A(t) - \mathcal A(s) \|_{\mathcal B(H^2(\mathbb R,\mathbb C^2),L^2(\mathbb R,\mathbb C^2)))} \leq 
\| \mathbf B(t) - \mathbf B(s) \|_{\mathbb C^{2\times 2}}
= \frac{|g(t)-g(s)|}{2\Omega_g^2} 
\leq  \frac{C |t-s)|}{2\Omega_g^2},
\end{equation*}
 since $t\mapsto g(t)$ is Lipschitz if $\boldsymbol\psi$ is differentiable with respect to $t$.
 \end{proof}
 
%---------------------------------------------------------------------------
 
\section{The essential spectrum of the monodromy operator}\label{sec:EssSpecEqual}

In this section we prove the main result, Theorem~\ref{thm:EssSpecEqual}, which gives
conditions under which
$\sigma_{\text{ess}}(\mathcal M) = \sigma_{\text{ess}}(\mathcal M_\infty)$.

\begin{proof}[Proof of Theorem~\ref{thm:EssSpecEqual}]
The lumped model we consider consists of fiber segments (single-mode fibers and a fiber amplifier) and discrete input-output devices (a dispersion compensation element, an output coupler, and a fast saturable absorber). We let $t\in[0,T]$ denote location in the laser loop.
In a fiber segment of length, $L$, that starts at location $t=T_1$, we have $t=t_{\rm loc} + T_1 \in [T_1,T_1+L]$, where $t_{\rm loc}$ denotes distance along the fiber. For an input-output device at location, $t$, we use $t_-$ and $t_+$ to denote the locations of the 
input and output to the device, and we impose the ordering $t_- < t_+$. We let $\mathcal U(t,s)$ and $\mathcal U_\infty(t,s)$, for $t>s$, denote the linearized evolution  and the asymptotic  linearized evolution operators
from location $s$ to location $t$. 
In particular, for an input-output device at location, $t$,  
the linearized transfer operator of the device is denoted by $\mathcal U(t_+,t_-)$.
The corresponding monodromy operators are then given by $\mathcal M = \mathcal U(T,0)$ and $\mathcal M_\infty = \mathcal U_\infty(T,0)$.
As in \eqref{eq:Linearization of the round trip operator}, $\mathcal M$ and $\mathcal M_\infty$ are both compositions of the linearized transfer operators of the fibers and devices in the lumped model. By Weyl's essential spectrum theorem~\cite{kato2013perturbation}, we just need to show that there is a compact operator, $\mathcal K$ so that 
\begin{equation}\label{eq:CompactPerturbation}
\mathcal M = \mathcal M_\infty + \mathcal K.
\end{equation}
To do so we will inductively 
show that at the location, $t$, of the end of each fiber segment that
\begin{equation}\label{eq:CompactPerturbationFiber}
\mathcal U(t,0) = \mathcal U_\infty(t,0) + \mathcal K(t),
\end{equation}
and that  at the exit, $t_+$, to each input-output device, that
\begin{equation}\label{eq:CompactPerturbationInOut}
\mathcal U(t_+,0) = \mathcal U_\infty(t_+,0) + \mathcal K(t_+),
\end{equation}
for some compact operators, $\mathcal K(t)$ and $\mathcal K(t_+)$.

First, we show that \eqref{eq:CompactPerturbationFiber} holds in the fiber amplifier.
The argument is the same for the single-mode fibers.
For a fiber segment of length, $L$, starting at location, $T_1$, an argument based on the variation of parameters formula (see \cite[Lemma 5.1]{zweck2021essential}) shows that,
for all $t\in[T_1,T_1+L]$,
\begin{equation}\label{eq:CompactPerturbationVofP}
\mathcal U(t,0) \,\,=\,\, \mathcal U_\infty(t,T_1) \circ \mathcal U(T_1,0) \,\,+\,\,
\int_{T_1}^t \mathcal U_\infty(t,t')\circ\mathbf M(t') \circ\mathcal U(t',0)\, dt',
\end{equation}
where $\mathbf M$ is the multiplication operator given by \eqref{eq:linearized operator of FA for evolution family new}.
Indeed, this equation is consistent at $t=T_1$ and implies that
\begin{equation}
\partial_t \,\mathcal U(t,0) = \mathcal L(t)\mathcal U(t,0).
\end{equation}

\begin{lemma}\label{lemma:IntegralCompact}
The operator 
\begin{equation}\label{eq:IntegralCompact}
\widetilde{\mathcal K}(t) = \int_{T_1}^t \mathcal U_\infty(t,t')\circ\mathbf M(t') \circ\mathcal U(t',0)\, dt'
\end{equation}
 is compact.
\end{lemma}

Given this lemma and substituting the induction hypothesis,
\begin{equation}\label{eq:CompactPerturbationIndHyp}
\mathcal U(T_1,0) = \mathcal U_\infty(T_1,0) + \mathcal K(T_1),
\end{equation}
into \eqref{eq:CompactPerturbationVofP} yields \eqref{eq:CompactPerturbationFiber} 
with
\begin{equation}
\mathcal K(t) \,\,=\,\, \mathcal U_\infty(t,T_1)\circ \mathcal K(T_1) 
\,\,+\,\, \widetilde{\mathcal K}(t),
\end{equation} 
which is compact since the composition of a bounded and a compact operator is compact.

Second, we show that \eqref{eq:CompactPerturbationInOut} holds for each input-output device.
Let 
\begin{equation}
\mathcal B(t_+,t_-) \,\,=\,\, \mathcal U(t_+,t_-) - \mathcal U_\infty(t_+,t_-).
\end{equation}
For all the input-output devices in the lumped model we are considering, except for the 
fast saturable absorber, $\mathcal B(t_+,t_-)=0$. By  \eqref{eq:SA linearized transfer function}, for the saturable absorber, 
$\mathcal B(t_+,t_-)(\boldsymbol{u}) = \mathbf B\boldsymbol{u}$ is a multiplication operator with
\begin{equation}\label{eq:BoperatorSA}
 \mathbf B(x) = (\ell_0 - \ell(\psi(x))\,\mathbf I  - \frac{2 \ell^2(\psi(x))}{\ell_0 P_{\rm sat}}
 \boldsymbol \psi  \boldsymbol \psi^T,
 \end{equation}
where 
\begin{equation}
    \ell(\psi) = \frac{\ell_0}{1 + |\psi_{\text{in}}|^2/P_{\text{sat}}}.
\end{equation}
Since $\boldsymbol\psi$ is assumed to be bounded, $\mathcal B(t_+,t_-)\in\mathcal B(L^2(\mathbb R,\mathbb C^2))$ is bounded but is not compact. Nevertheless,
we have the following theorem.

\begin{theorem}\label{thm:CompactPerturbationFSA}
Under the assumptions of Theorem~\ref{thm:EssSpecEqual}, for the fast saturable absorber
the operator, $\mathcal B(t_+ , t_-)\circ \mathcal U_\infty(t_- , 0)$, is compact.
\end{theorem}

Given this theorem and substituting the induction hypothesis,
\begin{equation}\label{eq:CompactPerturbationInOutIndHyp}
\mathcal U(t_-,0) = \mathcal U_\infty(t_-,0) + \mathcal K(t_-),
\end{equation}
into $\mathcal U(t_+,0) = \mathcal U(t_+,t_-)\circ \mathcal U(t_-,0)$ yields  
\eqref{eq:CompactPerturbationInOut} with
\begin{equation}
\mathcal K(t_+) \,\,=\,\, \mathcal B(t_+ , t_-)\circ \mathcal U_\infty(t_-,0) \,\,+\,\,
\mathcal U(t_+,t_-)\circ  \mathcal K(t_-),
\end{equation}
which is compact by Theorem~\ref{thm:CompactPerturbationFSA} and Proposition~\ref{prop:Existence of evolution family for SA}.
\end{proof}

\begin{proof}[Proof of Lemma~\ref{lemma:IntegralCompact}]
The proof uses the same basic ideas as in the proof of the analogous result for the complex Ginzburg-Landau equation given in~\cite[Theorem 5.1]{zweck2021essential}.
Here we confine our attention to showing that the integrand, $\mathcal C$, in \eqref{eq:IntegralCompact}
is compact. To do so, it suffices to show that the adjoint, $\mathcal C^*$, is compact.

Throughout the proof, we use times, $0<s<t<L$, that are local to the fiber, and we let
$\tau=L-t$ and $\sigma=L-s$ be the corresponding backwards time variables. 
Since the adjoint differential operator is defined by $\mathcal L^*(\tau) := [\mathcal L(L-\tau)]^*$, we have that
\begin{equation}
\mathcal L^*(\tau) = \mathcal L^*_\infty(\tau) + \mathbf M^*(L-\tau).
\end{equation}
By definition, the adjoint linearized evolution operator, $\mathcal U^*(\sigma,\tau)$,  in the fiber is the operator that satisfies
\begin{equation}
\partial_\sigma\, \mathcal U^*(\sigma,\tau) = \mathcal L^*(\sigma) \mathcal U^*(\sigma,\tau).
\end{equation}
This  operator is characterized by the equation
\begin{equation}
\langle \mathcal U(t,s)\boldsymbol{u}(s), \boldsymbol{v}(\tau) \rangle_{L^2(\mathbb R, \mathbb C^2)} \,\,=\,\,
\langle \boldsymbol{u}(s), \mathcal U^*(\sigma,\tau) \boldsymbol{v}(\tau) \rangle_{L^2(\mathbb R, \mathbb C^2)}.
\end{equation}
Therefore, 
\begin{equation}
[\mathcal U(t,s)]^* \,\,=\,\, \mathcal U^*(L-s,L-t).
\end{equation}
Letting $\tau'=L-t'$, we find that
\begin{equation}
\mathcal C^* \,\,=\,\, \mathcal U^*(L,\tau')\circ \mathbf M^*(L-\tau')\circ \mathcal U^*_\infty(\tau',\tau).
\end{equation}
As in Theorem~\ref{thm:EssSpecRCP}, 
$\mathcal L^*(\tau')$ is a relatively compact perturbation of
$\mathcal L^*_\infty(\tau')$. Therefore, there is a $\lambda(\tau')\in \rho(\mathcal L^*_\infty(\tau'))$ so that $\mathbf M^*(L-\tau') \circ (\mathcal L^*_\infty(\tau')-\lambda(\tau'))^{-1}$ is compact. Furthermore, by Theorem~\ref{FiberAmpAnalyticSemigroupThm} 
for the fiber amplifier (which also holds for the adjoint operators) and the corresponding result for the single mode fibers (modeled with the additional spectral filtering term as in~\eqref{eq:NLSFilter}, see \cite[Lemma 5.2]{zweck2021essential}), we have that
$(\mathcal L^*_\infty(\tau')-\lambda(\tau'))\circ \mathcal U^*_\infty(\tau',\tau)$
is bounded. 
Therefore, 
\begin{equation}
\mathcal C^* \,\,=\,\, \mathcal U^*(L,\tau')\circ \mathbf M^*(L-\tau')\circ 
(\mathcal L^*_\infty(\tau')-\lambda(\tau'))^{-1}\circ (\mathcal L^*_\infty(\tau')-\lambda(\tau')) \circ
\mathcal U^*_\infty(\tau',\tau).
\end{equation}
is compact, as required.
\end{proof}
 
 The proof of Theorem~\ref{thm:CompactPerturbationFSA} relies on the
 Kolmogorov-Riesz compactness theorem, which can be stated as 
 follows~\cite{hanche2010kolmogorov}.
 
 \begin{theorem}\label{KRCompactnessThm}
 A subset, $\mathfrak F \subset L^2(\mathbb R,\mathbb C^2) $,
is totally bounded  if and only if the following three conditions hold:
\begin{enumerate}
\item $\mathfrak F $ is bounded,
\item for all $\epsilon > 0$ there is an $R>0$ so that for all $f \in \mathfrak F$,
\begin{equation}
\int\limits_{|x|>R} \| f(x) \|^2_{\mathbb C^2}\, dx \,\, < \,\, \epsilon^2,\qquad\text{and}
\end{equation}
\item for all $\epsilon >0$ there is a  $\delta>0$ so that for all $f \in \mathfrak F$ and
$y\in\mathbb R$ with $|y|<\delta$,
\begin{equation}
\int\limits_{\mathbb R} \| f(x+y) -f(x)\|^2_{\mathbb C^2}\, dx \,\, < \,\, \epsilon^2.
\end{equation}
\end{enumerate}
 \end{theorem}

 \begin{proof}[Proof of Theorem~\ref{thm:CompactPerturbationFSA}]
 
 We first show that, at the input to the saturable absorber, 
\begin{equation}\label{eq:UptoSALtoH}
\mathcal U_\infty(t_-,0) \in \mathcal B(L^2(\mathbb R, \mathbb C^2), H^2(\mathbb R, \mathbb C^2)).
\end{equation}
This property holds since  the transfer operators for the fiber amplifier and  
the single-mode fibers with an additional spectral filtering term
satisfy
 \begin{equation}\label{eq:FABoundedLtoH}
 \mathcal U_\infty^{\rm FA}, \mathcal U_\infty^{\rm SMF} \in \mathcal B(L^2(\mathbb R, \mathbb C^2), H^2(\mathbb R, \mathbb C^2)),
 \end{equation}
 and  since \eqref{eq:DiscOpsBounded} holds for the DCF element and the output coupler. To establish \eqref{eq:FABoundedLtoH} for $ \mathcal U_\infty^{\rm FA}$, we use~\eqref{eq:FTFiberAmp} to obtain
 \begin{align}
\|\mathcal U_\infty^{\rm FA}\boldsymbol{u}\|_{H^2(\mathbb R,\mathbb C^2)}^2
&\leq C_1 \| (1+\omega^2) (\widehat{\mathcal U_\infty^{\rm FA}}\widehat{\boldsymbol{u}})(\omega) \|_{L^2(\mathbb R,\mathbb C^2)}^2
\\
&= C_1 \int_{\mathbb R} (1+\omega^2)^2 \exp\left( (1 -\omega^2/\Omega_g^2)G_{\rm FA} \right)\| \widehat {\boldsymbol{u}} (\omega) \|_{\mathbb C^2}^2\, d\omega
\\
& \leq C_2 \| \boldsymbol{u} \|_{L^2(\mathbb R,\mathbb C^2)}^2.
\end{align}
The proof for $\mathcal{U}_{\infty}^{\text{SMF}}$ is similar. 

From this point on, the proofs is analogous to the proof of \cite[Theorem 3.1]{zweck2021essential} that, for the complex Ginzburg-Landau equation,
 $\mathcal L(t)$ is a relatively compact perturbation of $\mathcal L_\infty$ 
There we showed that the operator $\mathbf M(t)\circ (\mathcal L_\infty - \lambda)^{-1}$
was compact using the exponential decay and weak regularity of $\boldsymbol\psi$
and the fact that $(\mathcal L_\infty - \lambda)^{-1}$ maps bounded sets in
$L^2(\mathbb R,\mathbb C^2)$ to bounded sets in $H^2(\mathbb R,\mathbb C^2)$ (endowed with the standard Sobolev norm). Here we show that 
$\mathcal K :=\mathcal B(t_+ , t_-)\circ \mathcal U_\infty(t_- , 0)$, is compact using the
exponential decay and weak regularity of $\boldsymbol\psi$ in the saturable absorber, together with~\eqref{eq:FABoundedLtoH}. Specifically, it suffices to show that
for any bounded family of functions, $\mathfrak H\subset L^2(\mathbb R,\mathbb C^2) $,
the subset $\mathfrak F = \mathcal K(\mathfrak H)\subset L^2(\mathbb R,\mathbb C^2) $
is totally bounded. To do so, we check the three conditions of the Kolmogorov-Riesz compactness Theorem~\ref{KRCompactnessThm}.

For the first condition, we observe that $\mathfrak F$ is bounded since the 
operator $\mathcal K$ and the subset $\mathfrak H$ are both bounded.
Let $\mathfrak G= \mathcal U_\infty(t_- , 0)(\mathfrak H)\subset H^2(\mathbb R, \mathbb C^2)$.
Since  $\mathfrak H$ is bounded, \eqref{eq:UptoSALtoH} implies that
\begin{equation}\label{eq:supG}
\sup\limits_{g\in\mathfrak G}\| g \|_{H^2(\mathbb R, \mathbb C^2)} < \infty.
\end{equation}

To verify the second condition, given $f\in\mathfrak F$, there is a  $g\in\mathfrak G$ so that
$f=\mathbf Bg$ where $\mathbf{B}$ is given by \eqref{eq:BoperatorSA}. Therefore,
\begin{equation}\label{eq:SACond2A}
\int\limits_{|x|>R} \| f(x)\|_{\mathbb C^2}^2\, dx 
\,\,\leq\,\, 
\int\limits_{|x|>R} \| \mathbf B(x)\|_{\mathbb C^{2\times 2}}^2 \|g(x)\|_{\mathbb C^2}^2 \, dx.
\end{equation}
Let $C_{\mathfrak G} = \sup\limits_{g\in \mathfrak G} \| g \|_{L^2(\mathbb R,\mathbb C^2)}$.
By Hypothesis~\ref{hyp:Conditions on the solution of SA}, 
  $\exists$ $R_1 >0$ so that
$\| \mathbf B(x)\|_{\mathbb C^{2\times 2}} < e^{-r|x|}/C_{\mathfrak G} $ for all $|x| > R_1$.
Therefore, if $R>R_1$, 
\begin{equation}\label{eq:SACond2B}
\int\limits_{|x|>R} \| f(x)\|_{\mathbb C^2}^2\, dx 
\,\,\leq\,\, \frac 1{C_{\mathfrak G} ^2}  e^{-2rR} \int_{|x|>R} \| g(x)\|^2_{\mathbb C^{2}}\, dx
\,\,\leq\,\, e^{-2rR} \,\,\leq\,\,\epsilon^2,
\end{equation}
provided  also that $R > |\log\epsilon|/r$.

For the third condition, we recall from Hypothesis~\ref{hyp:Conditions on the solution of SA} that $\mathbf B\in C^1(\mathbb R, \mathbb C^{2\times 2})$. 
Since $\mathfrak G\subset H^2(\mathbb R, \mathbb C^2)$, we know that
$\mathfrak F \subset H^1(\mathbb R, \mathbb C^2)$. By a result
in Evans~\cite[\S 5.8.2]{evans2010partial} on the difference quotient of a  $H^1$ function, we find that,
\begin{align}
\int_{\mathbb R} \| f(x+y) - f(x) \|^2_{\mathbb C^2} \, dx
 & \leq |y|^2 \, \| f_x \|^2_{L^2(\mathbb R,\mathbb C^2)} \nonumber \\
&   \leq |y|^2 \,\left[ \| \mathbf B_x g \|_{L^2(\mathbb R,\mathbb C^2)}
 + \| \mathbf B g_x \|_{L^2(\mathbb R,\mathbb C^2)} 
 \right]^2 \nonumber\\
 %% &
& \leq C|y|^2 \operatorname{max} \{ \| \mathbf B\|^2_{L^\infty(\mathbb R, \mathbb C^2)},
\| \mathbf B_x\|^2_{L^\infty(\mathbb R, \mathbb C^2)}
 \} \, \| g \|^2_{H^2(\mathbb R, \mathbb C^2)},\label{eq:SACond3}
\end{align}
for some constant, $C$. 
Finally, by Hypothesis~\ref{hyp:Conditions on the solution of SA} and \eqref{eq:supG}, the right hand side of \eqref{eq:SACond3} can be made arbitrarily small, provided $y$ is close enough to zero.
\end{proof}

%\section{Conclusions}
%\label{sec:conclusions}

\appendix
\section{Completion of Proof of Lemma~\ref{lem:F is C1}}\label{AppendixC1} 
To complete the proof we establish the estimates
in \eqref{eq:Bound on F' new} and \eqref{eq:Gestimates}.
%\begin{proof}
 By \eqref{eq:linearized operator of FA for evolution family A}, \eqref{eq:linearized operator of FA for evolution family B}, and \eqref{eq:linearized operator of FA for evolution family C}, 
    \begin{equation}
        \begin{aligned}
        & \norm{F(t+h) - F(t) -hF'(t)}_{L^2(\mathbb{R},\mathbb{C}^2)}\\ 
       & \leq  \norm{\{\textbf{B}(t+h) - \textbf{B}(t) - h\partial_t \textbf{B}(t)\} \partial_x^2 \boldsymbol{v}}_{L^2(\mathbb{R},\mathbb{C}^2)}\\
        &+ \norm{\{\widetilde{\textbf{M}}_1(t+h) - \widetilde{\textbf{M}}_1(t) - \partial_t \widetilde{\textbf{M}}_1(t)\} \boldsymbol{v}}_{L^2(\mathbb{R},\mathbb{C}^2)}\\
        &+ \lVert \boldsymbol{\phi}(t+h)\langle \boldsymbol{\psi}(t+h), \boldsymbol{v}\rangle - \boldsymbol{\phi}(t)\langle \boldsymbol{\psi}(t), \boldsymbol{v}\rangle - h\partial_t(\boldsymbol{\phi}(t)\langle \boldsymbol{\psi}(t), \boldsymbol{v}\rangle) \rVert_{L^2(\mathbb{R},\mathbb{C}^2)}.
        \end{aligned}
        \label{eq:bound on F'}
    \end{equation}
    
    To establish \eqref{eq:Bound on F' new} we estimate each of the term in \eqref{eq:bound on F'}. We estimate the first term in \eqref{eq:bound on F'} by    \begin{align*}
        & \norm{\{\textbf{B}(t+h) - \textbf{B}(t) - h\partial_t \textbf{B}(t)\} \partial_x^2 \boldsymbol{v}}_{L^2(\mathbb{R},\mathbb{C}^2)}^2\\
        \leq\; & \norm{\textbf{B}(t+h) - \textbf{B}(t) - h\partial_t \textbf{B}(t)}_{\mathbb{C}^{2\times 2}}^2 \int_{\mathbb{R}} \norm{\partial_x^2 \boldsymbol{v}(x)}_{\mathbb{C}^2}^2 dx\\
        \leq\; & \norm{\int_{t}^{t+h} \{ (\partial_t \textbf{B})(\tau) - (\partial_t \textbf{B})(t) \} d\tau }_{F}^2 \norm{\boldsymbol{v}}_{H^2(\mathbb{R},\mathbb{C}^2)}^2\\
        =\; & \sum_{i,j=1}^{2} \abs{\int_{t}^{t+h} \{ (\partial_t \textbf{B})_{ij}(\tau) - (\partial_t \textbf{B})_{ij}(t) \} d\tau }^2 \norm{\boldsymbol{v}}_{H^2(\mathbb{R},\mathbb{C}^2)}^2\\
        \leq\; & \sum_{i,j=1}^{2} h \int_{t}^{t+h} \abs{ (\partial_t \textbf{B})_{ij}(\tau) - (\partial_t \textbf{B})_{ij}(t) }^2 d\tau \quad \norm{\boldsymbol{v}}_{H^2(\mathbb{R},\mathbb{C}^2)}^2,
    \end{align*}
    where the last inequality follows from 
    \begin{equation}
        \abs{\int_a^b f(\tau) d\tau}^2 \leq (b-a) \int_a^b \abs{f(\tau)}^2 d\tau,
        \label{eq:Abs of integral}
    \end{equation}
    which is a special case of the Cauchy-Schwarz inequality. Consequently,
    \begin{equation}
    \begin{aligned}
        & \norm{\{\textbf{B}(t+h) - \textbf{B}(t) - h\partial_t \textbf{B}(t)\} \partial_x^2 \boldsymbol{v}}_{L^2(\mathbb{R},\mathbb{C}^2)}\\ 
        \leq\; & 2\sqrt{2}h \sup_{\tau \in (t,t+h)} \norm{(\partial_t \textbf{B})(\tau) - (\partial_t \textbf{B})(t)}_{\mathbb{C}^{2\times 2}} \quad \norm{\boldsymbol{v}}_{H^2(\mathbb{R},\mathbb{C}^2)}.
        \end{aligned}
        \label{eq:bound on B}
    \end{equation}
    Performing a similar calculation to estimate the second term in \eqref{eq:bound on F'}, we obtain
    \begin{align}
        &\norm{\{\widetilde{\textbf{M}}_1(t+h) - \widetilde{\textbf{M}}_1(t) - h\partial_t \widetilde{\textbf{M}}_1(t)\} \boldsymbol{v}}_{L^2(\mathbb{R},\mathbb{C}^2)}^2 
        \nonumber \\
            \leq\; & \int_{\mathbb{R}} \norm{\int_{t}^{t+h} \{ (\partial_t \widetilde{\textbf{M}}_1)(\tau,x) - (\partial_t \widetilde{\textbf{M}}_1)(t,x) \} d\tau }_{\mathbb{C}^{2\times 2}}^2 \norm{\boldsymbol{v}(x)}_{\mathbb{C}^2}^2 dx  \nonumber\\
        \leq\; & \sup_{x \in \mathbb{R}} \norm{\int_{t}^{t+h} \{ (\partial_t \widetilde{\textbf{M}}_1)(\tau,x) - (\partial_t \widetilde{\textbf{M}}_1)(t,x) \} d\tau }_{F}^2 \norm{\boldsymbol{v}}_{L^2(\mathbb{R},\mathbb{C}^2)}^2  \nonumber \\
        \leq\; & \sup_{x \in \mathbb{R}} \sum_{i,j=1}^{2} \abs{ \int_{t}^{t+h} \left\{  (\partial_t \widetilde{\textbf{M}}_1)_{ij}(\tau,x) - (\partial_t \widetilde{\textbf{M}}_1)_{ij}(t,x) \right\} d\tau }^2 \quad \norm{\boldsymbol{v}}_{H^2(\mathbb{R},\mathbb{C}^2)}^2  \nonumber\\
        \leq\; & \sup_{x \in \mathbb{R}} \sum_{i,j=1}^{2} h \int_{t}^{t+h} \abs{ (\partial_t \widetilde{\textbf{M}}_1)_{ij}(\tau,x) - (\partial_t \widetilde{\textbf{M}}_1)_{ij}(t,x) }^2 d\tau \quad \norm{\boldsymbol{v}}_{H^2(\mathbb{R},\mathbb{C}^2)}^2  \nonumber \\
        \leq\; & h^2 \sup_{(\tau,x) \in (t,t+h)\times \mathbb{R}} \norm{(\partial_t \widetilde{\textbf{M}}_1)(\tau,x) - (\partial_t \widetilde{\textbf{M}}_1)(t,x)}_F^2 \quad \norm{\boldsymbol{v}}_{H^2(\mathbb{R},\mathbb{C}^2)}^2  \nonumber\\
        \leq\; & 8h^2 \sup_{(\tau,x) \in (t,t+h)\times \mathbb{R}} \norm{(\partial_t \widetilde{\textbf{M}}_1)(\tau,x) - (\partial_t \widetilde{\textbf{M}}_1)(t,x)}_{\mathbb{C}^{2\times 2}}^2 \quad \norm{\boldsymbol{v}}_{H^2(\mathbb{R},\mathbb{C}^2)}^2.
           \label{eq:bound on M1}
    \end{align}
  
    Next, adding and subtracting $\boldsymbol{\phi}(t+h)\langle \boldsymbol{\psi}(t), \boldsymbol{v}\rangle$ in the third term of \eqref{eq:bound on F'}, we obtain
    \begin{equation}
    \begin{aligned}
        &\lVert \boldsymbol{\phi}(t+h)\langle \boldsymbol{\psi}(t+h), \boldsymbol{v}\rangle - \boldsymbol{\phi}(t)\langle \boldsymbol{\psi}(t), \boldsymbol{v}\rangle - h\partial_t(\boldsymbol{\phi}(t)\langle \boldsymbol{\psi}(t), \boldsymbol{v}\rangle) \rVert_{L^2(\mathbb{R},\mathbb{C}^2)}\\
       & \leq  \norm{\boldsymbol{\phi}(t+h) \langle \boldsymbol{\psi}(t+h) - \boldsymbol{\psi}(t), \boldsymbol{v} \rangle - \boldsymbol{\phi}(t)\langle h\partial_t \boldsymbol{\psi}(t), \boldsymbol{v} \rangle }_{L^2(\mathbb{R},\mathbb{C}^2)}\\
        &+ \norm{\{ \boldsymbol{\phi}(t+h) - \boldsymbol{\phi}(t) - h\partial_t\boldsymbol{\phi}(t) \} \langle \boldsymbol{\psi}(t), \boldsymbol{v}\rangle }_{L^2(\mathbb{R},\mathbb{C}^2)}.
    \end{aligned}
    \label{eq:bound on phi 1}
    \end{equation}
    Now, for any $\boldsymbol{u}$, $\boldsymbol{v}$, $\boldsymbol{w} \in L^2(\mathbb{R}, \mathbb{C}^2)$,
    \begin{equation}
        \norm{\boldsymbol{u} \langle \boldsymbol{v}, \boldsymbol{w} \rangle}_{L^2(\mathbb{R}, \mathbb{C}^2)}  \leq \norm{\boldsymbol{u}}_{L^2(\mathbb{R}, \mathbb{C}^2)} \; \norm{\boldsymbol{v}}_{L^2(\mathbb{R}, \mathbb{C}^2)} \; \norm{\boldsymbol{w}}_{L^2(\mathbb{R}, \mathbb{C}^2)}.
        \label{eq:Inequality for uvw}
    \end{equation}
    To estimate the first term in \eqref{eq:bound on phi 1}, we add and subtract $\boldsymbol{\phi}(t+h) \langle h \partial_t \boldsymbol{\psi}(t), \boldsymbol{v} \rangle$ and use \eqref{eq:Inequality for uvw} to obtain
    \begin{equation}
    \begin{aligned}
        &\norm{\boldsymbol{\phi}(t+h) \langle \boldsymbol{\psi}(t+h) - \boldsymbol{\psi}(t), \boldsymbol{v} \rangle - \boldsymbol{\phi}(t)\langle h\partial_t \boldsymbol{\psi}(t), \boldsymbol{v} \rangle }_{L^2(\mathbb{R},\mathbb{C}^2)}\\
        \leq\; & \bigg\{ \norm{\boldsymbol{\phi}(t+h)}_{L^2(\mathbb{R},\mathbb{C}^2)} \norm{\boldsymbol{\psi}(t+h) - \boldsymbol{\psi}(t) - h \partial_t \boldsymbol{\psi}(t)}_{L^2(\mathbb{R},\mathbb{C}^2)}\\
        &+ \norm{\boldsymbol{\phi}(t+h) - \boldsymbol{\phi}(t)}_{L^2(\mathbb{R},\mathbb{C}^2)} \norm{h\partial_t \boldsymbol{\psi}(t)}_{L^2(\mathbb{R},\mathbb{C}^2)} \bigg\} \norm{\boldsymbol{v}}_{L^2(\mathbb{R},\mathbb{C}^2)}\\
        =\; & \bigg\{\norm{\boldsymbol{\phi}(t+h)}_{L^2(\mathbb{R},\mathbb{C}^2)} \norm{\int_{t}^{t+h} \{ (\partial_t \boldsymbol{\psi})(\tau) - (\partial_t \boldsymbol{\psi})(t) \} d\tau}_{L^2(\mathbb{R},\mathbb{C}^2)}\\
        &+ h \norm{\int_{t}^{t+h} (\partial_t \phi)(\tau) d\tau}_{L^2(\mathbb{R},\mathbb{C}^2)} \norm{\partial_t \boldsymbol{\psi}(t)}_{L^2(\mathbb{R},\mathbb{C}^2)} \bigg\} \norm{\boldsymbol{v}}_{H^2(\mathbb{R},\mathbb{C}^2)}\\
        \leq\; & \bigg\{ h \norm{\boldsymbol{\phi}(t+h)}_{L^2(\mathbb{R},\mathbb{C}^2)} \sup_{\tau \in (t,t+h)} \norm{(\partial_t \boldsymbol{\psi})(\tau) - (\partial_t \boldsymbol{\psi})(t)}_{L^2(\mathbb{R},\mathbb{C}^2)}\\
        &+ h^2 \sup_{\tau \in (t,t+h)} \norm{(\partial_t \boldsymbol{\phi})(\tau)}_{L^2(\mathbb{R},\mathbb{C}^2)} \quad \norm{\partial_t \boldsymbol{\psi}(t)}_{L^2(\mathbb{R},\mathbb{C}^2)} \bigg\} \norm{\boldsymbol{v}}_{H^2(\mathbb{R},\mathbb{C}^2)}.
    \end{aligned}
    \label{eq:bound on phi 2}
    \end{equation}
    Now,
    \begin{equation}
        \norm{\boldsymbol{\phi} (t)}_{L^2(\mathbb{R},\mathbb{C}^2)} 
    \leq \frac{g_0C}{E_{\text{sat}}} \norm{\boldsymbol{\psi}(t)}_{H^2(\mathbb{R},\mathbb{C}^2)},
        \label{eq:bound on phi 3}
    \end{equation}
    and
    \begin{equation}
    \begin{aligned}
        \norm{\partial_t \boldsymbol{\phi} (t)}_{L^2(\mathbb{R},\mathbb{C}^2)} 
                \leq & \frac{1}{g_0 E_{\text{sat}}} \bigg\{ \abs{\frac{-2}{E_{\text{sat}}}g^3(t)E^{\prime}(t)} \; \norm{\left( \boldsymbol{\psi}(t) + \frac{\partial_x^2 \boldsymbol{\psi}(t)}{\Omega_g^2}\right)}_{L^2(\mathbb{R},\mathbb{C}^2)}\\
        &+ g^2(t) \norm{\partial_t \left( \boldsymbol{\psi}(t) + \frac{\partial_x^2 \boldsymbol{\psi}(t)}{\Omega_g^2}\right)}_{L^2(\mathbb{R},\mathbb{C}^2)} \bigg\}\\
        \leq & \frac{2g_0^2 C}{E_{\text{sat}}^2}|E^{\prime}(t)| \; \norm{\boldsymbol{\psi}(t)}_{H^2(\mathbb{R},\mathbb{C}^2)} + \frac{2g_0 C}{E_{\text{sat}}} \norm{\partial_t \boldsymbol{\psi}(t)}_{H^2(\mathbb{R},\mathbb{C}^2)}.
    \end{aligned}
    \label{eq:bound on phi 4}
    \end{equation}
    Substituting \eqref{eq:bound on phi 3} and \eqref{eq:bound on phi 4} in \eqref{eq:bound on phi 2}, we obtain
    \begin{equation}
        \begin{aligned}
        &\norm{\boldsymbol{\phi}(t+h) \langle \boldsymbol{\psi}(t+h) - \boldsymbol{\psi}(t), \boldsymbol{v} \rangle - \boldsymbol{\phi}(t)\langle h\partial_t \boldsymbol{\psi}(t), \boldsymbol{v} \rangle }_{L^2(\mathbb{R},\mathbb{C}^2)}\\
        \leq\; & \bigg\{ \frac{g_0 hC}{E_{\text{sat}}} \norm{\boldsymbol{\psi}(t+h)}_{H^2(\mathbb{R},\mathbb{C}^2)} \sup_{\tau \in (t,t+h)} \norm{(\partial_t \boldsymbol{\psi})(\tau) - (\partial_t \boldsymbol{\psi})(t)}_{H^2(\mathbb{R},\mathbb{C}^2)}\\
        &+ \frac{2g_0^2 h^2 C}{E_{\text{sat}}^2} \sup_{\tau \in (t,t+h)} \abs{E^{\prime}(\tau)} \norm{\boldsymbol{\psi}(\tau)}_{H^2(\mathbb{R},\mathbb{C}^2)} \norm{\partial_t \boldsymbol{\psi}(t)}_{H^2(\mathbb{R},\mathbb{C}^2)}\\
        &+ \frac{2g_0 h^2 C}{E_{\text{sat}}} \sup_{\tau \in (t,t+h)} \norm{\partial_t \boldsymbol{\psi}(\tau)}_{H^2(\mathbb{R},\mathbb{C}^2)} \norm{\partial_t \boldsymbol{\psi}(t)}_{H^2(\mathbb{R},\mathbb{C}^2)}
        \bigg\} \norm{\boldsymbol{v}}_{H^2(\mathbb{R},\mathbb{C}^2)}.\\
        \end{aligned}
        \label{eq:bound on phi 6}
    \end{equation}
    Next to estimate the second term in \eqref{eq:bound on phi 1} we use \eqref{eq:Inequality for uvw} to obtain
    \begin{equation}
        \begin{aligned}
        &\norm{\{ \boldsymbol{\phi}(t+h) - \boldsymbol{\phi}(t) - h\partial_t\boldsymbol{\phi}(t) \} \langle \boldsymbol{\psi}(t), \boldsymbol{v}\rangle }_{L^2(\mathbb{R},\mathbb{C}^2)}\\
        \leq & \norm{\boldsymbol{\phi}(t+h) - \boldsymbol{\phi}(t) - h\partial_t\boldsymbol{\phi}(t)}_{L^2(\mathbb{R},\mathbb{C}^2)} \norm{\boldsymbol{\psi}(t)}_{L^2(\mathbb{R},\mathbb{C}^2)} \norm{\boldsymbol{v}}_{L^2(\mathbb{R},\mathbb{C}^2)},
        \end{aligned}
        \label{eq:bound on phi 7}
    \end{equation}
    and observe that, by \eqref{eq:Abs of integral} and Fubini's theorem, 
    \begin{equation}
    \begin{aligned}
        \norm{\boldsymbol{\phi}(t+h) - \boldsymbol{\phi}(t) - h\partial_t \boldsymbol{\phi}(t)}_{L^2(\mathbb{R},\mathbb{C}^2)}^2 &= \norm{\int_t^{t+h} \left( (\partial_t \boldsymbol{\phi})(\tau) - (\partial_t \boldsymbol{\phi})(t) \right) d\tau}_{L^2(\mathbb{R},\mathbb{C}^2)}^2\\
        \leq& h \int_t^{t+h} \norm{(\partial_t \boldsymbol{\phi})(\tau) - (\partial_t \boldsymbol{\phi})(t)}_{L^2(\mathbb{R},\mathbb{C}^2)}^2 d\tau\\
        \leq& h^2 \sup_{\tau \in (t,t+h)} \norm{(\partial_t \boldsymbol{\phi})(\tau) - (\partial_t \boldsymbol{\phi})(t)}_{L^2(\mathbb{R},\mathbb{C}^2)}^2.
        \label{eq:bound on phi 8}
    \end{aligned}
    \end{equation}
    Finally, substituting \eqref{eq:bound on B}, \eqref{eq:bound on M1}, \eqref{eq:bound on phi 7}, and \eqref{eq:bound on phi 8} in \eqref{eq:bound on F'}, we obtain \eqref{eq:Bound on F' new}.
%\qed
%\end{proof}

\section*{Acknowledgments}
We thank Yuri Latushkin and Curtis Menyuk for generously sharing their
expertise with us.
\StartRed
We thank the anonymous reviewers for their suggestions to improve the paper. 
\EndRed

\bibliographystyle{siamplain}
\bibliography{VSTheory}

\end{document}